\documentclass[journal]{IEEEtran}
\usepackage{amsmath,amssymb,bm}
\usepackage{booktabs}
\usepackage{graphicx}
\usepackage[caption=false,font=footnotesize]{subfig}
\usepackage{tikz}
\usepackage{circuitikz}
\usetikzlibrary{arrows.meta,positioning,fit}
\usepackage{cite}
\usepackage{placeins}
\usepackage[hidelinks]{hyperref}
\graphicspath{{figures/}}
\newcommand{\vect}[1]{\bm{#1}}
\newcommand{\mat}[1]{\bm{#1}}
\newcommand{\CoE}{\mathrm{Co}_{\mathrm e}}
\newcommand{\CoC}{\mathrm{Co}_{\mathrm c}}
\newtheorem{theorem}{Theorem}

\begin{document}
\title{Row-Local Spectral Certificates and Analog Courant Metrics for
Finite-Bandwidth Feedback-Coupled Solvers}
\author{Arash~Ghasemi%
\thanks{A. Ghasemi is with Sumer Analog LLC, Knoxville, TN 37923 USA, and
also with the University of Tennessee at Chattanooga, Chattanooga, TN 37403
USA.}}
\maketitle

\begin{abstract}
This paper identifies a dynamical constraint on analog-computing approaches in which a row of the matrix is represented by an impedance network. It shows that the fastest normalized mode is no more than 2$\pi$ times the largest combined unity-gain bandwidth (CUGBW) among all the circuit rows. The CUGBW of a row equals its finite-gain-adjusted unity-gain bandwidth plus the contributions of all rows coupled to it. Each contribution is the square root of the product of the two rows' unity-gain bandwidths multiplied by their coupling conductance and divided by the square root of the product of their total conductance loadings. This bound plays a role analogous to the Courant-number restriction in time-stepping methods by limiting the operator rates that analog hardware can physically represent and resolve at its outputs. It is shown that in mixed-signal approach, the Courant metrics become stability criteria, and they must be less than 2 in order for the solver to remain stable. The theory is validated using large-scale LTspice simulations across architectures ranging from CMOS to thermionic vacuum-tube circuits. The benchmark circuits implement a one-dimensional heat equation, a graph-based semi-supervised learning problem, and a graph-regularized regression.
\end{abstract}


\section{Introduction}
\IEEEPARstart{F}{eedback-coupled} analog networks solve structured linear
systems by allowing circuit dynamics to converge to a fixed point
\cite{huang2016,sun2019matrix,hasler2021,mannocci2021universal}. Their
computational latency is not determined by the gain-bandwidth product of an
isolated amplifier alone. Finite amplifier bandwidth interacts with row
loading and network connectivity, producing a spectrum of physical decay
rates. The lowest rate determines worst-case asymptotic settling while the highest
rate determines the shortest reduced-model time scale. If a sampled
controller advances the network state explicitly, the maximum rate governs the maximum stable update
interval.

The basic dynamics of analog linear solvers have been studied extensively. Sun
\emph{et al.} related a one-pole closed-loop resistive solver to the smallest
eigenvalue, logarithmic tolerance-settling time, and a finite-difference
spectral-radius condition \cite{sun2020complexity}. Hasler and Natarajan
analyzed configurable continuous-time systems with heterogeneous capacitance
and conductance matrices \cite{hasler2021}. Wang \emph{et al.} incorporated
finite gain bandwidth through a quadratic pole model for memristive linear
regression \cite{wang2021optimization}, while Mannocci and Ielmini developed a
general finite-gain, single-pole block-matrix model and Hurwitz conditions
\cite{mannocci2023generalized}. Open-loop matrix--vector readout has a
different loading-dependent time scale \cite{sun2021mvm}, and recent
architectures implement iterative or integrated closed-loop solvers
\cite{li2025iteration,zuo2025matrix,mannocci2026closedloop}. 

This paper instead addresses the complementary circuit-design question of
whether the fastest generalized mode can be certified from programmed row
parameters before a global eigensolution is performed. For reciprocal
follower networks with nonnegative passive couplings, the answer is a
row-local upper certificate whose terms are available from conductances, row
sums, dc gains, and unity-gain bandwidths. Among the closest models above,
none provides this particular hardware-parametric fast-rate certificate
together with its continuous-versus-sampled interpretation and bounded-input
rail conditions.

The paper makes four main contributions.
\begin{enumerate}
\item A loaded finite-bandwidth follower row is reduced to an induced
capacitance \(C_j=S_j/(2\pi\nu_{t,j})\). The realizable class consists of
reciprocal Laplacian-plus-anchor systems, which are Stieltjes/M-matrix systems.
\item A row-local certificate \(\lambda_{\mathrm{G,max}}\) is derived for the
upper generalized rate. It requires no eigendecomposition and becomes
asymptotically tight, for example, for the uniform 1D heat-transfer problem.
\item The exact and certified analog Courant numbers,
\(\CoE=T_s\lambda_{\max}\) and
\(\CoC=T_s\lambda_{\mathrm{G,max}}\), are introduced. They describe
observation resolution for a passive continuous solver, but become stability
quantities only for an explicit sampled update. A conservative rail-safe
invariant set is obtained in the same weighted norm.
\item Controlled modal LTspice tests directly excite the two spectral
endpoints. Vendor-macromodel studies then examine regular and irregular
200-node CMOS-based networks and a 30-node hybrid thermionic vacuum-tube implementation.
\end{enumerate}

\section{Finite-Bandwidth Row Model}

\subsection{Loaded Follower Cell}
Fig.~\ref{fig:row-cell} shows row \(j\). The reciprocal branch between rows
\(i\) and \(j\) has conductance \(g_{ji}=g_{ij}=1/R_{ji}\); \(g_{a,j}\)
is an optional grounding or data-fidelity anchor, not the unity-feedback path
of the voltage follower. The programmed source current is \(I_j\), the
summing-node voltage is \(v_j^\star\), and \(v_j\) is the represented
unknown.

\begin{figure}[!t]
\centering
\resizebox{\columnwidth}{!}{%
\begin{tikzpicture}[>=Latex,font=\scriptsize,
 resistor/.style={draw,minimum width=0.82cm,minimum height=0.22cm,inner sep=0pt},
 via/.style={circle,draw,fill=white,minimum size=3.6pt,inner sep=0pt}]
 \draw[line width=0.7pt] (-5.25,-2.75) rectangle (-4.62,2.22);
 \node[rotate=90,font=\normalsize] at (-4.94,-0.27) {Set gains};
 \coordinate (s) at (0,0);
 \fill (s) circle (1.8pt);
 \node[below right] at (s) {$v_j^\star$};
 \foreach \y/\vin/\rlabel in {
 1.45/{v_1}/{R_{j1}},0.73/{v_2}/{R_{j2}},
 -0.73/{v_i}/{R_{ji}},-1.45/{v_m}/{R_{jm}}}{
  \node[left] at (-3.55,\y) {$\vin$};
  \node[via] at (-3.18,\y) {};
  \draw[-Latex] (-3.05,\y)--(-2.43,\y);
  \node[resistor] at (-2.02,\y) {};
  \node[fill=white,inner sep=0.4pt] at (-2.02,\y+0.34) {$\rlabel$};
  \draw (-1.61,\y)--(-0.58,\y);
  \draw[dashed,-Latex] (-4.62,\y-0.25)--(-2.48,\y-0.25)--(-2.48,\y-0.10);
 }
 \draw (-0.58,1.45)--(-0.58,-1.45);
 \draw (-0.58,0)--(s);
 \node at (-3.34,0) {$\vdots$};
 \draw (s)--(1.05,0);
 \draw (1.05,-0.90)--(1.05,0.90)--(2.80,0)--cycle;
 \node at (1.28,0.25) {$+$};
 \node at (1.28,-0.60) {$-$};
 \draw (2.80,0)--(4.15,0);
 \node[via] at (4.15,0) {};
 \node[right] at (4.30,0) {$v_j$};
 \draw (1.05,-0.60)--(0.67,-0.60)--(0.67,-1.05)--(3.68,-1.05)--(3.68,0);
 \node[resistor,rotate=90] at (0,1.05) {};
 \node[left] at (-0.18,1.05) {$R_{a,j}$};
 \draw (0,1.46)--(0,1.82);
 \draw (-0.28,1.82)--(0.28,1.82);
 \draw (-0.20,1.95)--(0.20,1.95);
 \draw (-0.12,2.07)--(0.12,2.07);
 \draw (0,0.64)--(s);
 \draw[dashed,-Latex] (-4.62,2.08)--(-0.35,2.08)--(-0.35,1.05)--(-0.17,1.05);
 \draw (-0.58,0)--(-0.58,-1.82);
 \draw (-0.58,-2.13) circle (0.31);
 \draw[-Latex] (-0.58,-2.32)--(-0.58,-1.95);
 \node[left] at (-0.91,-2.05) {$I_j$};
 \draw (-0.58,-2.44)--(-0.58,-2.70);
 \draw (-0.86,-2.70)--(-0.30,-2.70);
 \draw (-0.78,-2.82)--(-0.38,-2.82);
 \draw (-0.70,-2.93)--(-0.46,-2.93);
 \draw[dashed,-Latex] (-4.62,-2.13)--(-0.94,-2.13);
 \draw (s)--(0,-1.78);
 \draw (-0.30,-1.78)--(0.30,-1.78);
 \draw (-0.30,-1.94)--(0.30,-1.94);
 \draw (0,-1.94)--(0,-2.39);
 \draw (-0.26,-2.39)--(0.26,-2.39);
 \draw (-0.18,-2.51)--(0.18,-2.51);
 \draw (-0.10,-2.62)--(0.10,-2.62);
 \node[right] at (0.35,-1.86) {$C_j$};
\end{tikzpicture}}
\caption{One programmable row of the reciprocal solver. \(R_{a,j}\) is an
anchoring branch, and \(C_j\) is the equivalent capacitance induced by the
finite-bandwidth follower. It is not a physical capacitor.}
\label{fig:row-cell}
\end{figure}
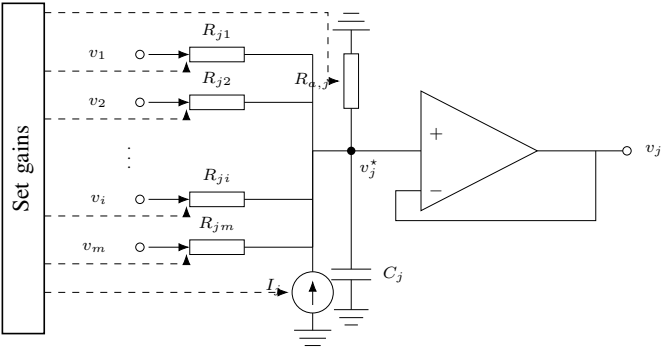

Kirchhoff's current law and the total loading are
\begin{equation}
 g_{a,j}v_j^\star-I_j+
 \sum_{i\ne j}g_{ji}(v_j^\star-v_i)=0,\qquad
 S_j=g_{a,j}+\sum_{i\ne j}g_{ji}.
\label{eq:kcl}
\end{equation}
Let \(A_{0j}\) and \(\nu_{t,j}\) be the follower's dc open-loop gain and
unity-gain bandwidth, and define
\(\alpha_j=A_{0j}/(1+A_{0j})\). After any slew-limited startup, the
single-pole small-signal relation is
\begin{equation}
 v_j^\star=\frac{\dot v_j}{2\pi\nu_{t,j}}+\frac{v_j}{\alpha_j}.
\label{eq:follower}
\end{equation}
Substitution in \eqref{eq:kcl} yields
\begin{equation}
 C_j\dot v_j+\frac{S_j}{\alpha_j}v_j
 -\sum_{i\ne j}g_{ji}v_i=I_j,\qquad
 C_j=\frac{S_j}{2\pi\nu_{t,j}}.
\label{eq:row-dynamics}
\end{equation}
Thus capacitance is jointly set by bandwidth and row loading, and hence, it is not an
independent tuning parameter.

\subsection{Generalized Operator and Scope}
Collecting the rows gives
\begin{equation}
 \mat C\dot{\vect v}+\mat G\vect v=\vect I,\quad
 G_{jj}=\frac{S_j}{\alpha_j},\quad G_{ji}=-g_{ji},
\label{eq:matrix-dynamics}
\end{equation}
where \(\mat C=\operatorname{diag}(C_j)\). Reciprocal couplings make
\(\mat G\) symmetric. With at least one effective anchor in every connected
component, \(\mat G\succ0\). The passive topology therefore implements a
Laplacian-plus-anchor Stieltjes matrix, rather than an arbitrary signed or
nonsymmetric operator. Finite gain increases the diagonal and makes the
steady solution \(\mat G^{-1}\vect I\) differ slightly from the intended
infinite-gain solution.

With \(\vect w=\mat C^{1/2}\vect v\),
\begin{equation}
 \dot{\vect w}+\mat A\vect w=\mat C^{-1/2}\vect I,\qquad
 \mat A=\mat C^{-1/2}\mat G\mat C^{-1/2}.
\label{eq:normalized}
\end{equation}
The symmetric positive-definite matrix
\(\mat A=\mat Q\mat\Lambda\mat Q^{T}\) has rates
\(0<\lambda_{\min}\le\lambda_k\le\lambda_{\max}\).
For constant \(\vect I\), the error satisfies
\begin{equation}
 \|\vect v(t)-\vect v_\infty\|_{\mat C}
 \le e^{-\lambda_{\min}t}
 \|\vect v(0)-\vect v_\infty\|_{\mat C}.
\label{eq:settling}
\end{equation}
Hence a worst-case relative reduction \(\epsilon\) requires
\(t_\epsilon\ge\lambda_{\min}^{-1}\ln(1/\epsilon)\), consistent with
the lower-edge settling analyses in \cite{sun2020complexity,hasler2021}.
The ratio
\(\kappa_{\mathrm{dyn}}=\lambda_{\max}/\lambda_{\min}\) measures the
separation of the fastest and slowest physical time scales.

\section{Row-Local Upper-Rate Certificate}

The entries of \(\mat A\) follow directly from \eqref{eq:row-dynamics} and are
\begin{equation}
 a_{jj}=\frac{2\pi\nu_{t,j}}{\alpha_j},\qquad
 a_{ji}=-\frac{2\pi g_{ji}\sqrt{\nu_{t,j}\nu_{t,i}}}
 {\sqrt{S_jS_i}}\quad(i\ne j).
\label{eq:A-entries}
\end{equation}

\begin{theorem}[Row-local spectral certificate]
For the reciprocal positive-definite network in
\eqref{eq:matrix-dynamics}, define
\begin{align}
 L_{\mathrm G}
 &=\min_j\!\left[
 \frac{2\pi\nu_{t,j}}{\alpha_j}
 -\sum_{i\ne j}
 \frac{2\pi g_{ji}\sqrt{\nu_{t,j}\nu_{t,i}}}{\sqrt{S_jS_i}}
 \right],\label{eq:LG}\\
 \lambda_{\mathrm{G,max}}
 &=\max_j\!\left[
 \frac{2\pi\nu_{t,j}}{\alpha_j}
 +\sum_{i\ne j}
 \frac{2\pi g_{ji}\sqrt{\nu_{t,j}\nu_{t,i}}}{\sqrt{S_jS_i}}
 \right].\label{eq:UG}
\end{align}
Then \(L_{\mathrm G}\le\lambda_{\min}\) and
\(\lambda_{\max}\le\lambda_{\mathrm{G,max}}\).
\end{theorem}

\begin{IEEEproof}
Every eigenvalue of the real symmetric matrix \(\mat A\) lies in a
Gershgorin interval centered at \(a_{jj}\) with radius
\(\sum_{i\ne j}|a_{ji}|\) \cite{gershgorin1931}. Substituting
\eqref{eq:A-entries} and taking the smallest left and largest right endpoints
gives \eqref{eq:LG}--\eqref{eq:UG}.
\end{IEEEproof}

The lower expression is useful only when it is positive. A nonpositive value
is compatible with \(\mat A\succ0\), but provides no useful certified settling
rate. The upper expression remains finite and certifies the fastest
reduced-model mode using only local row quantities.
If a desired certified dimensionless rate is \(\CoC\le\gamma\), every row can
be checked independently against
\begin{equation}
 T_s\!\left[
 \frac{2\pi\nu_{t,j}}{\alpha_j}
 +\sum_{i\ne j}
 \frac{2\pi g_{ji}\sqrt{\nu_{t,j}\nu_{t,i}}}{\sqrt{S_jS_i}}
 \right]\le\gamma .
\label{eq:row-design}
\end{equation}

Multiplying every \(g_{ji}\), \(g_{a,j}\), and \(I_j\) by a common positive
factor leaves the algebraic solution unchanged. It also scales \(\mat G\) and
\(\mat C\) together, so \(\mat A\), its spectrum, and
\(\lambda_{\mathrm{G,max}}\) are invariant. The lower resistance can still
change power, noise, slew, output loading, and nonlinearity, but it does not by
itself evade the small-signal spectral rate.

\subsection{A PVT-Robust Local Margin}
The same row structure provides a conservative process, voltage, and
temperature (PVT) extension without corner-by-corner eigen-solutions. Suppose
characterization supplies simultaneous bounds
\[
\begin{gathered}
 0<\alpha_j^{-}\le\alpha_j,\qquad
 \nu_{t,j}\le\nu_{t,j}^{+},\\
 0\le g_{ji}\le g_{ji}^{+},\qquad
 0<S_j^{-}\le S_j .
\end{gathered}
\]
Each term on the right-hand side of \eqref{eq:UG} is then no greater than the
corresponding worst-case term, so
\begin{equation}
 \lambda_{\max}\le\overline{\lambda}_{\mathrm{G,max}}
 :=\max_j\!\left[
 \frac{2\pi\nu_{t,j}^{+}}{\alpha_j^{-}}+
 \sum_{i\ne j}
 \frac{2\pi g_{ji}^{+}
 \sqrt{\nu_{t,j}^{+}\nu_{t,i}^{+}}}
 {\sqrt{S_j^{-}S_i^{-}}}
 \right].
\label{eq:robust-UG}
\end{equation}
When \(S_j\) is itself assembled from uncertain branches, \(S_j^{-}\) must be
formed from simultaneous lower bounds on its anchor and incident
conductances. Replacing \(\lambda_{\mathrm{G,max}}\) by
\(\overline{\lambda}_{\mathrm{G,max}}\) in the sampled conditions below gives
a row-computable PVT margin. The construction is conservative by design and
retains the advantage that only local characterization data are needed.

\section{Continuous Safety and Sampled Updates}

\subsection{Conservative Continuous-Time Rail Certificate}
Define the modal forcing
\(\vect b(t)=\mat Q^T\mat C^{-1/2}\vect I(t)\), assume
\(\|\vect b(t)\|_2\le\bar b\), and set
\(r_{\mathrm{safe}}=\min_j\sqrt{C_j}V_{\mathrm{safe},j}\), where each
\(V_{\mathrm{safe},j}\) lies strictly within both the supply rails and the
small-signal region. Variation of constants gives
\begin{equation}
 \|\mat C^{1/2}\vect v(t)\|_2
 \le e^{-\lambda_{\min}t}\|\mat C^{1/2}\vect v(0)\|_2
 +(1-e^{-\lambda_{\min}t})\frac{\bar b}{\lambda_{\min}}.
\label{eq:rail-bound}
\end{equation}
Consequently,
\begin{equation}
 \max\!\left\{\|\mat C^{1/2}\vect v(0)\|_2,
 \frac{\bar b}{\lambda_{\min}}\right\}<r_{\mathrm{safe}}
\label{eq:rail-safe}
\end{equation}
is sufficient for \(|v_j(t)|<V_{\mathrm{safe},j}\) for all \(j,t\).
This \(C\)-weighted ball is inscribed in the rectangular rail region and can
be highly conservative. Failure of \eqref{eq:rail-safe} does not imply rail
contact. In fact, it indicates passing close to the rails, which may strengthen nonlinear instability/inaccuracy effects. 
\subsection{Explicit Mixed-Signal Update}
If a controller samples the rows and explicitly advances the physical state,
the implementable forward-Euler update is
\begin{equation}
 \vect v[n+1]=\vect v[n]+T_s\mat C^{-1}
 \bigl(\vect I[n]-\mat G\vect v[n]\bigr).
\label{eq:euler-physical}
\end{equation}
In modal coordinates its homogeneous factors are
\(1-T_s\lambda_k\). Define the circuit Courant numbers: 
\begin{equation}
 \CoE=T_s\lambda_{\max},\qquad
 \CoC=T_s\lambda_{\mathrm{G,max}},\qquad \CoE\le\CoC .
\label{eq:courant}
\end{equation}
These are CFL(Courant–Friedrichs–Lewy)-inspired circuit metrics since the underlying CFL and
forward-Euler results are classical \cite{courant1928}. Zero-input
asymptotic stability is exactly \(\CoE<2\); \(\CoC<2\) is a row-computable
sufficient condition, while \(\CoC\le1\) also prevents alternating modal
signs. For \(\|\vect b[n]\|_2\le\bar b\), the rail-safe ball is invariant if
\begin{equation}
 \|\mat C^{1/2}\vect v[0]\|_2<r_{\mathrm{safe}},\quad
 \lambda_{\min}>\frac{\bar b}{r_{\mathrm{safe}}},\quad
 \lambda_{\max}<2f_s-\frac{\bar b}{r_{\mathrm{safe}}}.
\label{eq:euler-rail}
\end{equation}
The sampled limits apply only to \eqref{eq:euler-physical}. A passively
observed continuous solver has no selected integration step and therefore no
CFL stability condition applies. There, \(T_s\lambda_{\max}\) only compares the
acquisition interval with the fastest reduced-model time scale.

Table~\ref{tab:roles} summarizes which spectral edge answers each design
question. A positive \(L_{\mathrm G}\) can replace \(\lambda_{\min}\) in a
sufficient settling or rail test. When \(L_{\mathrm G}\le0\), the local lower
test is simply silent. In contrast, \(\lambda_{\mathrm{G,max}}\) is always a
valid finite upper certificate under the stated assumptions.

\begin{table}[!t]
\centering
\caption{Spectral Quantities Used in Circuit Design}
\label{tab:roles}
\footnotesize
\setlength{\tabcolsep}{2.5pt}
\begin{tabular}{p{.30\columnwidth}p{.29\columnwidth}p{.31\columnwidth}}
\toprule
Question & Exact condition & Row-local sufficient test\\
\midrule
Residual \(\le\epsilon\) &
\(t\ge\ln(1/\epsilon)/\lambda_{\min}\) &
\(t\ge\ln(1/\epsilon)/L_{\mathrm G}\), \(L_{\mathrm G}>0\)\\
Euler stability &
\(T_s\lambda_{\max}<2\) &
\(T_s\lambda_{\mathrm{G,max}}<2\)\\
No modal alternation &
\(T_s\lambda_{\max}\le1\) &
\(T_s\lambda_{\mathrm{G,max}}\le1\)\\
PVT Euler stability &
corner spectrum &
\(T_s\overline{\lambda}_{\mathrm{G,max}}<2\)\\
\bottomrule
\end{tabular}
\end{table}

\section{Uniform Heat-Transfer Problem Network Corollary}

For an \(n\)-node one-dimensional Dirichlet heat stencil, assign a common
conductance \(g_h\) to each of the two row branches, a common
\((A_0,\nu_t)\), and \(\alpha=A_0/(1+A_0)\). Then
\(S_j=2g_h\), \(C_h=g_h/(\pi\nu_t)\), and the exact generalized rates are
\begin{equation}
 \lambda_k=\pi\nu_t\!\left(\frac{2}{\alpha}
 -2\cos\frac{k\pi}{n+1}\right),\qquad k=1,\ldots,n .
\label{eq:heat-spectrum}
\end{equation}
Thus
\begin{align}
 \lambda_{\min}
 &=\frac{2\pi\nu_t}{A_0}
 +2\pi\nu_t\!\left(1-\cos\frac{\pi}{n+1}\right),\label{eq:heat-min}\\
 \lambda_{\max}
 &=\pi\nu_t\!\left(\frac{2}{\alpha}
 +2\cos\frac{\pi}{n+1}\right),\label{eq:heat-max}\\
 \lambda_{\mathrm{G,max}}
 &=2\pi\nu_t\!\left(\frac{1}{\alpha}+1\right).\label{eq:heat-cert}
\end{align}
When \(A_0\gg2(n+1)^2/\pi^2\),
\(\lambda_{\min}\sim\pi^3\nu_t/(n+1)^2\), whereas
\(\lambda_{\max}\sim4\pi\nu_t\). The slow time therefore grows as \(n^2\),
but the fastest rate stays at the amplifier-bandwidth scale. Moreover,
\(\lambda_{\mathrm{G,max}}-\lambda_{\max}
=2\pi\nu_t[1-\cos(\pi/(n+1))]\), proving asymptotic tightness of the
upper certificate. The scale \(g_h\) cancels exactly. At fixed finite gain,
the apparent lower-rate floor \(2\pi\nu_t/A_0\) is inseparable from
finite-gain steady-state bias.

For example, the \(n=200\) LTC1151 parameter set in
Table~\ref{tab:spectrum} gives
\(\lambda_{\max}=31.414\,\mathrm{M\,s^{-1}}\) and
\(\lambda_{\mathrm{G,max}}=31.416\,\mathrm{M\,s^{-1}}\). An explicit update
is therefore certified stable for \(f_s>15.708\,\mathrm{MHz}\), but the
stronger nonalternating certificate requires
\(f_s\ge31.416\,\mathrm{MHz}\). At \(f_s=20\,\mathrm{MHz}\), the update is
stable but its highest modes alternate; at \(15\,\mathrm{MHz}\), the exact
fastest factor has magnitude greater than one.

\section{Validation}

\subsection{Method and Direct Endpoint Test}
Three distinct forms of evidence are utilized in this paper. First, a controlled
behavioral LTspice network tests the reduced one-pole equations directly.
Second, exact eigensolutions of the assembled \((\mat G,\mat C)\) pairs test
certificate conservatism. Third, scalable circuits are created by utilizing a vendor-provided hybrid macromodel/transistor-level SPICE model for the op-amps. 

The controlled circuit uses \(n=20\), \(R=2\,\mathrm{k}\Omega\),
\(A_0=10^7\), and \(\nu_t=2.5\,\mathrm{MHz}\). Its internal states are
initialized to the exact sine eigenvectors \(k=1\) and \(k=20\), separately,
so the measured node contains no intentional modal mixture. Exponential fits
give \(1.75446722\times10^5\,\mathrm{s}^{-1}\) for the slow mode and
\(3.12405085\times10^7\,\mathrm{s}^{-1}\) for the fast mode, versus exact
values \(1.75446542\times10^5\) and \(3.12404831\times10^7\,\mathrm{s}^{-1}\).
The relative errors are \(1.03\times10^{-6}\) and \(8.13\times10^{-7}\).

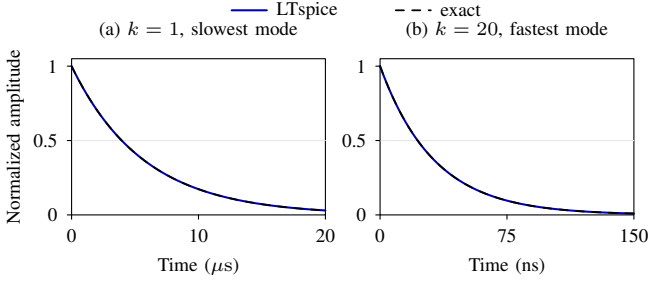
\begin{figure}[!t]
\centering
\resizebox{\columnwidth}{!}{
\begin{tikzpicture}[font=\scriptsize,line cap=round,line join=round]
\draw[blue!70!black,line width=0.75pt] (2.05,2.62) -- (2.50,2.62);
\node[right] at (2.50,2.62) {LTspice};
\draw[black,dashed,line width=0.65pt] (4.15,2.62) -- (4.60,2.62);
\node[right] at (4.60,2.62) {exact};
\begin{scope}[xshift=0.00cm]
\draw[black] (0,0) rectangle (3.25,2.00);
\draw[gray!20] (0,0.95238) -- (3.25,0.95238);
\draw (0,0.00000) -- (-0.055,0.00000) node[left] {0};
\draw (0,0.95238) -- (-0.055,0.95238) node[left] {0.5};
\draw (0,1.90476) -- (-0.055,1.90476) node[left] {1};
\draw (0.00000,0) -- (0.00000,-0.055) node[below] {0};
\draw (1.62500,0) -- (1.62500,-0.055) node[below] {10};
\draw (3.25000,0) -- (3.25000,-0.055) node[below] {20};
\draw[blue!70!black,line width=0.75pt] plot coordinates {(0.00000,1.90476) (0.04063,1.82302) (0.08125,1.74479) (0.12188,1.66992) (0.16250,1.59825) (0.20312,1.52967) (0.24375,1.46402) (0.28438,1.40120) (0.32500,1.34107) (0.36562,1.28352) (0.40625,1.22844) (0.44688,1.17572) (0.48750,1.12526) (0.52813,1.07697) (0.56875,1.03076) (0.60938,0.98652) (0.65000,0.94419) (0.69063,0.90367) (0.73125,0.86489) (0.77187,0.82777) (0.81250,0.79225) (0.85313,0.75825) (0.89375,0.72571) (0.93437,0.69457) (0.97500,0.66476) (1.01562,0.63624) (1.05625,0.60893) (1.09688,0.58280) (1.13750,0.55779) (1.17813,0.53386) (1.21875,0.51095) (1.25937,0.48902) (1.30000,0.46803) (1.34062,0.44795) (1.38125,0.42873) (1.42188,0.41033) (1.46250,0.39272) (1.50313,0.37587) (1.54375,0.35974) (1.58438,0.34430) (1.62500,0.32952) (1.66563,0.31538) (1.70625,0.30185) (1.74687,0.28889) (1.78750,0.27650) (1.82812,0.26463) (1.86875,0.25328) (1.90938,0.24241) (1.95000,0.23200) (1.99063,0.22205) (2.03125,0.21252) (2.07188,0.20340) (2.11250,0.19467) (2.15313,0.18632) (2.19375,0.17832) (2.23438,0.17067) (2.27500,0.16334) (2.31563,0.15633) (2.35625,0.14963) (2.39688,0.14320) (2.43750,0.13706) (2.47813,0.13118) (2.51875,0.12555) (2.55938,0.12016) (2.60000,0.11500) (2.64063,0.11007) (2.68125,0.10535) (2.72187,0.10082) (2.76250,0.09650) (2.80313,0.09236) (2.84375,0.08839) (2.88437,0.08460) (2.92500,0.08097) (2.96563,0.07749) (3.00625,0.07417) (3.04688,0.07099) (3.08750,0.06794) (3.12812,0.06502) (3.16875,0.06223) (3.20938,0.05956) (3.25000,0.05701)};
\draw[black,dashed,line width=0.65pt] plot coordinates {(0.00000,1.90476) (0.04063,1.82302) (0.08125,1.74479) (0.12188,1.66991) (0.16250,1.59825) (0.20312,1.52967) (0.24375,1.46402) (0.28438,1.40120) (0.32500,1.34106) (0.36562,1.28351) (0.40625,1.22843) (0.44688,1.17572) (0.48750,1.12526) (0.52813,1.07697) (0.56875,1.03076) (0.60938,0.98652) (0.65000,0.94419) (0.69063,0.90367) (0.73125,0.86489) (0.77187,0.82777) (0.81250,0.79225) (0.85313,0.75825) (0.89375,0.72571) (0.93437,0.69457) (0.97500,0.66476) (1.01562,0.63624) (1.05625,0.60893) (1.09688,0.58280) (1.13750,0.55779) (1.17813,0.53386) (1.21875,0.51095) (1.25937,0.48902) (1.30000,0.46803) (1.34062,0.44795) (1.38125,0.42873) (1.42188,0.41033) (1.46250,0.39272) (1.50313,0.37587) (1.54375,0.35974) (1.58438,0.34430) (1.62500,0.32952) (1.66563,0.31538) (1.70625,0.30185) (1.74687,0.28889) (1.78750,0.27650) (1.82812,0.26463) (1.86875,0.25328) (1.90938,0.24241) (1.95000,0.23200) (1.99063,0.22205) (2.03125,0.21252) (2.07188,0.20340) (2.11250,0.19467) (2.15313,0.18632) (2.19375,0.17832) (2.23438,0.17067) (2.27500,0.16334) (2.31563,0.15633) (2.35625,0.14963) (2.39688,0.14320) (2.43750,0.13706) (2.47813,0.13118) (2.51875,0.12555) (2.55938,0.12016) (2.60000,0.11500) (2.64063,0.11007) (2.68125,0.10535) (2.72187,0.10082) (2.76250,0.09650) (2.80313,0.09236) (2.84375,0.08839) (2.88437,0.08460) (2.92500,0.08097) (2.96563,0.07749) (3.00625,0.07417) (3.04688,0.07099) (3.08750,0.06794) (3.12812,0.06502) (3.16875,0.06223) (3.20938,0.05956) (3.25000,0.05701)};
\node[below=0.42cm] at (1.62500,0) {Time ($\mu\mathrm{s}$)};
\node[above=0.14cm] at (1.62500,2.00000) {(a) $k=1$, slowest mode};
\node[rotate=90] at (-0.72,1.00000) {Normalized amplitude};
\end{scope}
\begin{scope}[xshift=3.95cm]
\draw[black] (0,0) rectangle (3.25,2.00);
\draw[gray!20] (0,0.95238) -- (3.25,0.95238);
\draw (0,0.00000) -- (-0.055,0.00000) node[left] {0};
\draw (0,0.95238) -- (-0.055,0.95238) node[left] {0.5};
\draw (0,1.90476) -- (-0.055,1.90476) node[left] {1};
\draw (0.00000,0) -- (0.00000,-0.055) node[below] {0};
\draw (1.62500,0) -- (1.62500,-0.055) node[below] {75};
\draw (3.25000,0) -- (3.25000,-0.055) node[below] {150};
\draw[blue!70!black,line width=0.75pt] plot coordinates {(0.00000,1.90476) (0.04062,1.79640) (0.08125,1.69419) (0.12187,1.59780) (0.16250,1.50690) (0.20312,1.42117) (0.24375,1.34031) (0.28437,1.26406) (0.32500,1.19214) (0.36562,1.12432) (0.40625,1.06035) (0.44687,1.00002) (0.48750,0.94313) (0.52812,0.88947) (0.56875,0.83886) (0.60938,0.79114) (0.65000,0.74613) (0.69063,0.70368) (0.73125,0.66364) (0.77187,0.62589) (0.81250,0.59028) (0.85313,0.55670) (0.89375,0.52502) (0.93437,0.49515) (0.97500,0.46698) (1.01562,0.44041) (1.05625,0.41536) (1.09688,0.39173) (1.13750,0.36944) (1.17812,0.34842) (1.21875,0.32860) (1.25937,0.30990) (1.30000,0.29227) (1.34062,0.27564) (1.38125,0.25996) (1.42188,0.24517) (1.46250,0.23122) (1.50313,0.21807) (1.54375,0.20566) (1.58438,0.19396) (1.62500,0.18292) (1.66562,0.17252) (1.70625,0.16270) (1.74687,0.15345) (1.78750,0.14472) (1.82812,0.13648) (1.86875,0.12872) (1.90937,0.12139) (1.95000,0.11449) (1.99063,0.10797) (2.03125,0.10183) (2.07187,0.09604) (2.11250,0.09057) (2.15313,0.08542) (2.19375,0.08056) (2.23437,0.07598) (2.27500,0.07165) (2.31562,0.06758) (2.35625,0.06373) (2.39688,0.06011) (2.43750,0.05669) (2.47812,0.05346) (2.51875,0.05042) (2.55938,0.04755) (2.60000,0.04485) (2.64062,0.04230) (2.68125,0.03989) (2.72188,0.03762) (2.76250,0.03548) (2.80312,0.03346) (2.84375,0.03156) (2.88437,0.02976) (2.92500,0.02807) (2.96563,0.02647) (3.00625,0.02497) (3.04687,0.02355) (3.08750,0.02221) (3.12812,0.02094) (3.16875,0.01975) (3.20938,0.01863) (3.25000,0.01757)};
\draw[black,dashed,line width=0.65pt] plot coordinates {(0.00000,1.90476) (0.04062,1.79639) (0.08125,1.69419) (0.12187,1.59780) (0.16250,1.50690) (0.20312,1.42117) (0.24375,1.34031) (0.28437,1.26406) (0.32500,1.19214) (0.36562,1.12431) (0.40625,1.06035) (0.44687,1.00002) (0.48750,0.94313) (0.52812,0.88947) (0.56875,0.83886) (0.60938,0.79114) (0.65000,0.74613) (0.69063,0.70368) (0.73125,0.66364) (0.77187,0.62589) (0.81250,0.59028) (0.85313,0.55670) (0.89375,0.52502) (0.93437,0.49515) (0.97500,0.46698) (1.01562,0.44041) (1.05625,0.41536) (1.09688,0.39173) (1.13750,0.36944) (1.17812,0.34842) (1.21875,0.32860) (1.25937,0.30990) (1.30000,0.29227) (1.34062,0.27564) (1.38125,0.25996) (1.42188,0.24517) (1.46250,0.23122) (1.50313,0.21807) (1.54375,0.20566) (1.58438,0.19396) (1.62500,0.18292) (1.66562,0.17252) (1.70625,0.16270) (1.74687,0.15345) (1.78750,0.14472) (1.82812,0.13648) (1.86875,0.12872) (1.90937,0.12139) (1.95000,0.11449) (1.99063,0.10797) (2.03125,0.10183) (2.07187,0.09604) (2.11250,0.09057) (2.15313,0.08542) (2.19375,0.08056) (2.23437,0.07598) (2.27500,0.07165) (2.31562,0.06758) (2.35625,0.06373) (2.39688,0.06011) (2.43750,0.05669) (2.47812,0.05346) (2.51875,0.05042) (2.55938,0.04755) (2.60000,0.04485) (2.64062,0.04230) (2.68125,0.03989) (2.72188,0.03762) (2.76250,0.03548) (2.80312,0.03346) (2.84375,0.03156) (2.88437,0.02976) (2.92500,0.02807) (2.96563,0.02647) (3.00625,0.02497) (3.04687,0.02355) (3.08750,0.02221) (3.12812,0.02094) (3.16875,0.01975) (3.20938,0.01863) (3.25000,0.01757)};
\node[below=0.42cm] at (1.62500,0) {Time (ns)};
\node[above=0.14cm] at (1.62500,2.00000) {(b) $k=20$, fastest mode};
\end{scope}
\end{tikzpicture}}
\caption{Controlled LTspice endpoint test. Each trace is normalized to its
first post-initial-condition sample; dashed curves are
\(\exp(-\lambda_k t)\) using \eqref{eq:heat-spectrum}. (a) Slow eigenvector
\(k=1\). (b) Fast eigenvector \(k=20\).}
\label{fig:modal-validation}
\end{figure}

\subsection{Heat Transfer Problem Across CMOS and Vacuum-Tube Technologies}
The main heat-transfer benchmark has 200 follower rows, equal
\(2\,\mathrm{k}\Omega\) branches, end currents
\(\pm0.1\,\mathrm{mA}\), and \(\pm2.5\,\mathrm V\) supplies. Its ideal
profile is \(v_j=0.2-0.4j/201\) V. LTC1151, ADA4528-1, and OPA388 vendor
macromodels realize the same network but have distinct startup trajectories.
The final maximum errors are \(0.040\,\mathrm{mV}\) at \(100\,\mathrm{ms}\),
\(3.84\,\mathrm{mV}\) at \(100\,\mathrm{ms}\), and
\(1.72\,\mathrm{mV}\) at \(34.1\,\mathrm{ms}\), respectively. Their
datasheet \(A_0,\nu_t\) values \cite{ltc1151datasheet,ada4528datasheet,
opa388datasheet} are used only for the reduced rates in
Table~\ref{tab:spectrum}.

A separate 30-node network uses chopper-stabilized Philbrick K2-W/K2-P
macros operated as followers. Small-signal extraction gives
\(|A_0|=2.26\times10^5\) and mean
\(\nu_t=1.048\,\mathrm{MHz}\). The resulting surrogate has
\(\lambda_{\max}=1.314\times10^7\,\mathrm{s}^{-1}\) and
\(\lambda_{\mathrm{G,max}}=1.32\times10^7\,\mathrm{s}^{-1}\). The full
hybrid thermionic/behavioral transient reaches a \(2.08\,\mathrm{mV}\)
maximum profile error at \(300\,\mathrm{ms}\), but its slow nonmonotone startup
reflects chopper and macro states outside the one-pole surrogate. The assembled
hybrid macro and its tube-level K2-W core are documented in
Appendix~\ref{app:k2w-macro}.
The corresponding CMOS steady profiles and complete hybrid transient are shown
in Fig.~\ref{fig:heat-validation}.

\subsection{Irregular Graph Workloads}
Graph-based semisupervised learning represents labeled and unlabeled samples as vertices in a weighted graph. In this case, nearby or similar samples are connected by large weights, so the learned score should vary smoothly over the graph \cite{zhu2003,smola2003}. Recent resistive-memory graph accelerators implement graph embeddings through recurrent analog MVMs \cite{wang2023esgnn}. The graph-labeling circuit solves
\((\mat L+\eta\mat M)\vect f=\eta\mat M\vect y\) for 200 two-moons samples,
ten labels, a symmetrized \(15\)-nearest-neighbor graph, and \(\eta=4\)
\cite{zhu2003,smola2003}. All 200 output signs become correct after
\(1.29\,\mu\mathrm{s}\), and the final classification is 200/200 with minimum
margin \(0.637\,\mathrm V\). The maximum final deviation from the ideal
linear solve is \(67.7\,\mathrm{mV}\), showing that task correctness can
precede high-accuracy fixed-point convergence. The distance to the final
circuit state falls below \(1\,\mathrm{mV}\) at \(45.8\,\mu\mathrm{s}\).

The 13-band regression network solves
\((\mat I+0.9\mat L_{\mathrm{band}})\vect f=\vect y\), with six graph
neighbors on each side. It is more strongly anchored and has
\(\kappa_{\mathrm{dyn}}=7.975\). Its final maximum error relative to the ideal
solve is \(3.94\times10^{-7}\,\mathrm V\), and the distance to the final
circuit state falls below \(1\,\mathrm{mV}\) at \(3.49\,\mu\mathrm{s}\).
These cases also show the expected dependence on topology.
\(\lambda_{\mathrm{G,max}}/\lambda_{\max}\) is 1.672 for the irregular
label graph and 1.561 for the banded network, versus 1.00006 for the regular
heat-transfer stencil.
Fig.~\ref{fig:graph-validation} shows the classified samples, graph-error
histories, representative regression transients, and final regression profile.

As a separate lower-edge check, an effective tail rate can be extracted from
two persistent error thresholds,
\(\widehat{\lambda}_{\mathrm{tail}}
=\ln(\delta_{\mathrm{hi}}/\delta_{\mathrm{lo}})
/(t_{\mathrm{lo}}-t_{\mathrm{hi}})\). The LTC1151 heat thresholds at
4 and \(1\,\mathrm{mV}\) give
\(1.87\times10^3\,\mathrm{s}^{-1}=0.976\lambda_{\min}\). The graph-labeling
10-to-\(1\,\mathrm{mV}\) interval gives
\(1.17\times10^5\,\mathrm{s}^{-1}=0.700\lambda_{\min}\), while the
13-band interval gives
\(2.71\times10^6\,\mathrm{s}^{-1}=1.14\lambda_{\min}\). These values are
consistent with mixed modal amplitudes and threshold selection. By contrast,
the OPA388 4-to-\(2\,\mathrm{mV}\) interval gives only
\(0.258\lambda_{\min}\), directly exposing additional macromodel dynamics.
The application transients therefore provide order-of-magnitude support for
the lower edge and validate fixed-point behavior. The isolated experiment in
Fig.~\ref{fig:modal-validation}, rather than the overview and close-up in
Figs.~\ref{fig:heat-transients} and \ref{fig:heat-closeup}, is the direct
spectral-endpoint test.
The fixed-point errors and tolerance or task times for all six application
cases are collected in Table~\ref{tab:application}.

\begin{figure}[!ht]
\centering
\subfloat[LTC1151]{%
\includegraphics[width=.42\columnwidth]{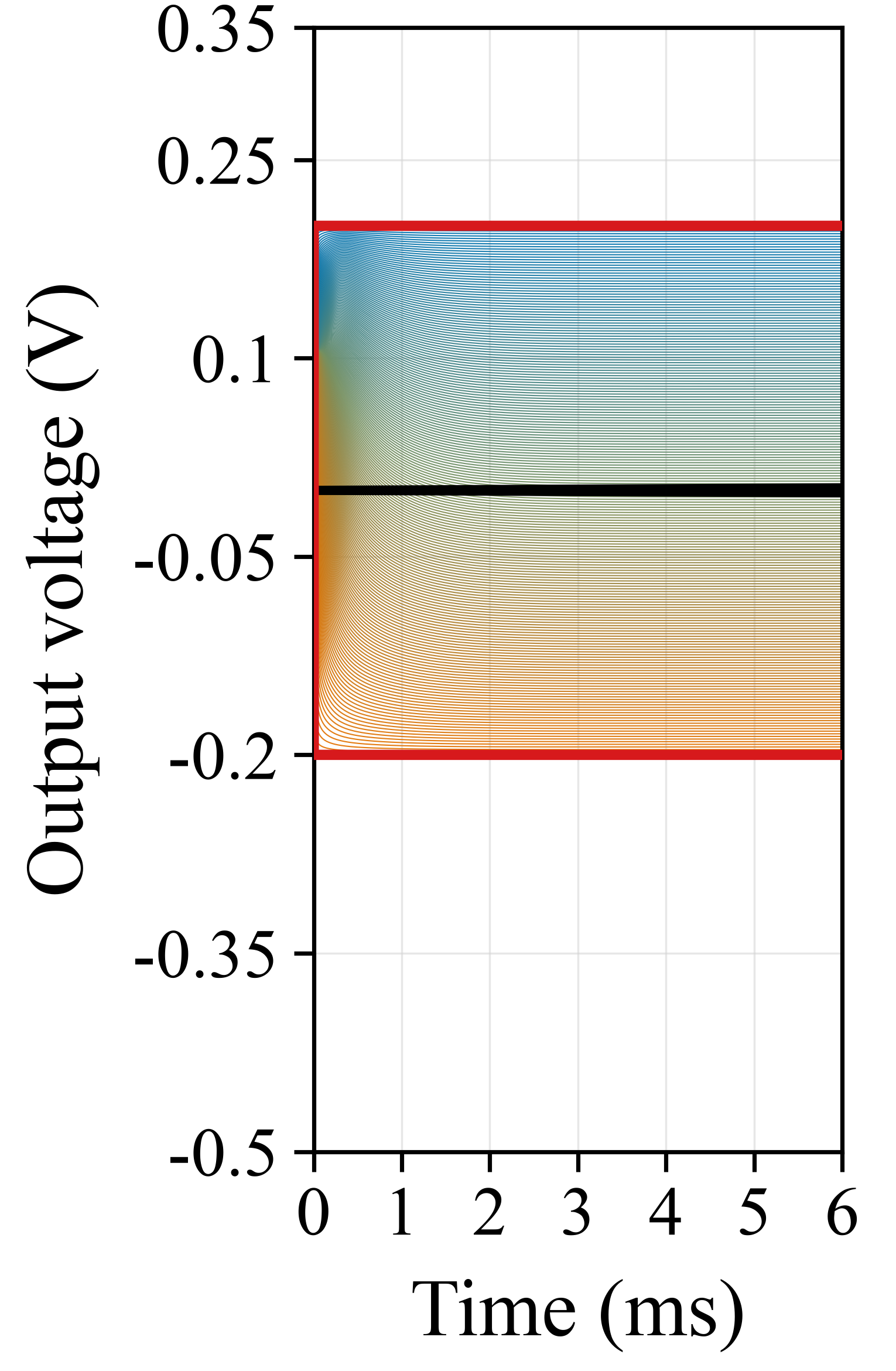}}
\hspace{.004\columnwidth}
\subfloat[ADA4528-1]{%
\includegraphics[width=.42\columnwidth]{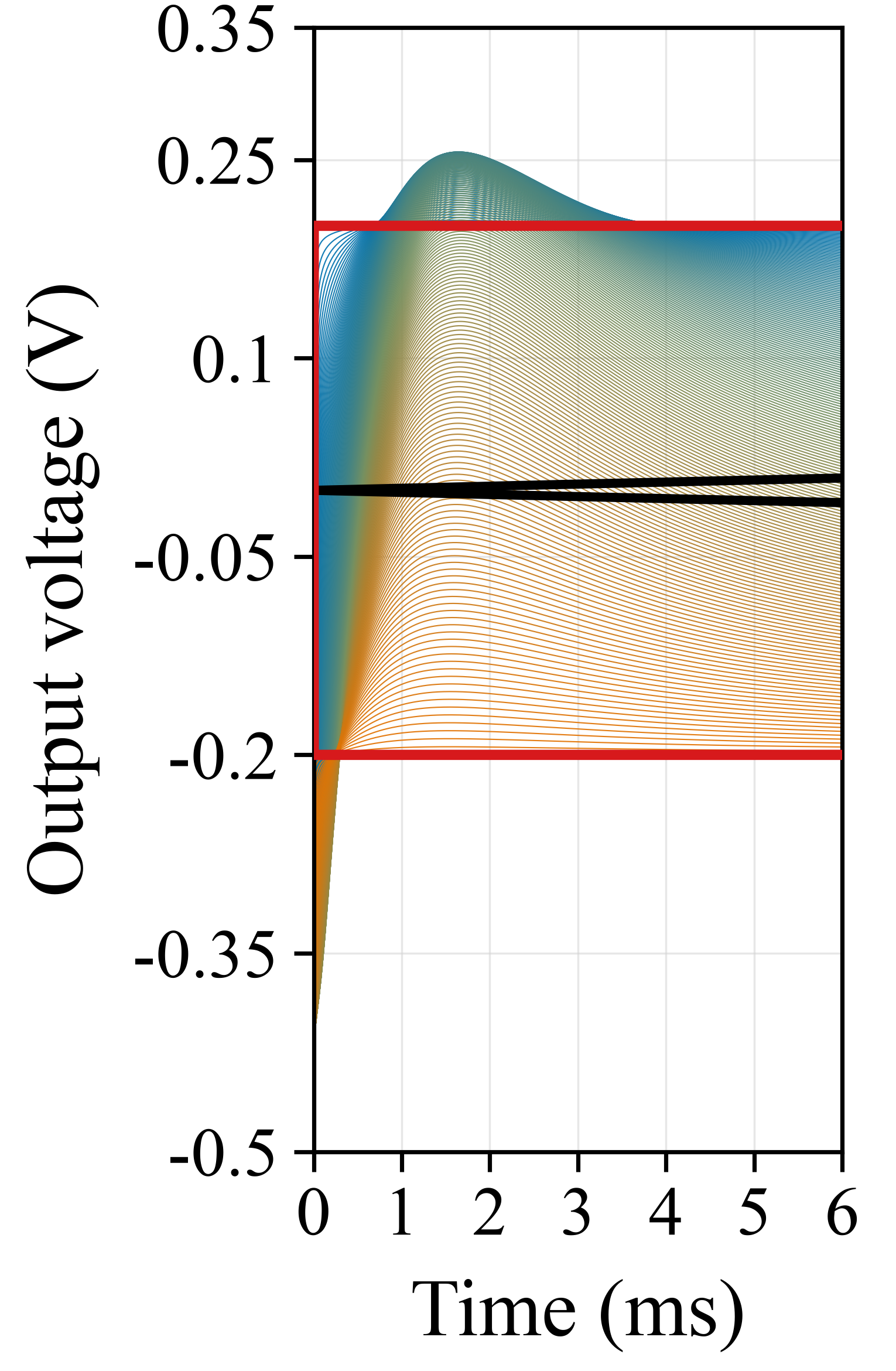}}
\hspace{.002\columnwidth}
\raisebox{.38in}{%
\includegraphics[height=1.44in]{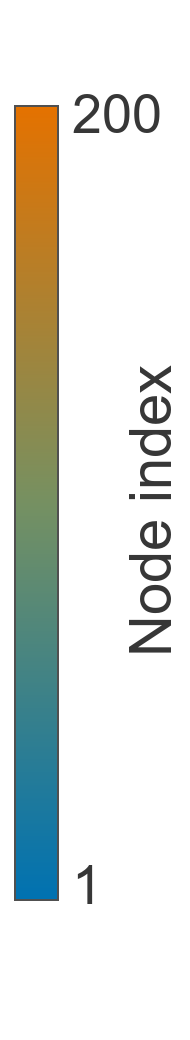}}
\caption{LTspice heat-transfer transients of all 200 outputs over
\(6\,\mathrm{ms}\). Red and black curves use
\(\lambda_{\mathrm{G,max}}\) and \(L_{\mathrm G}\), respectively, as modal
rate references rather than pointwise voltage bounds.}
\label{fig:heat-transients}
\par\vspace{.12em}
\makebox[\columnwidth][c]{%
\begin{minipage}[t]{.7\columnwidth}
\vspace{0pt}
\centering
\includegraphics[width=\linewidth]{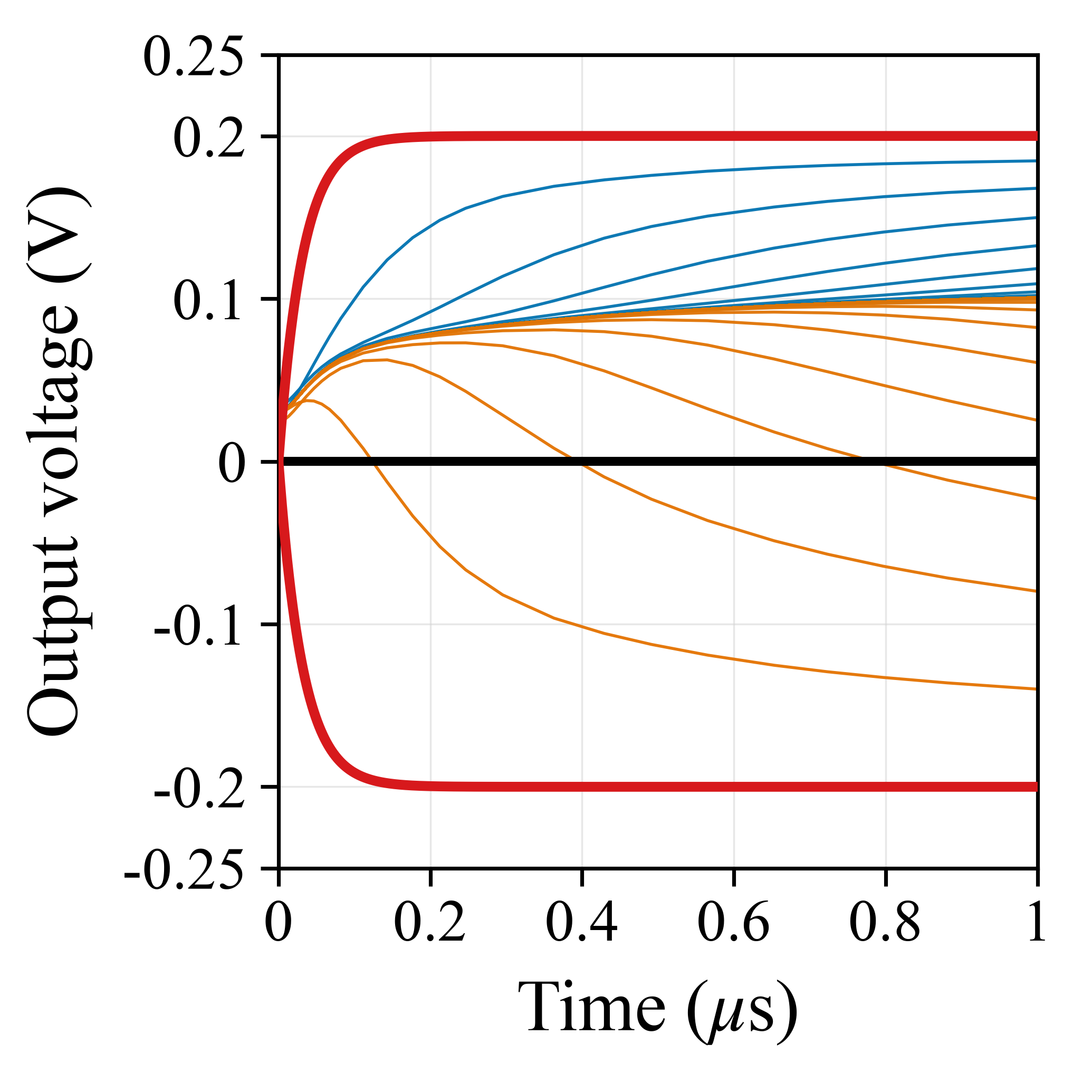}
\end{minipage}%
\hspace{.006\columnwidth}%
\begin{minipage}[t]{.10\columnwidth}
\vspace{.10in}
\centering
\includegraphics[height=1.55in]{fig_transient_node_index_colorbar.png}
\end{minipage}%
}
\caption{LTC1151 transient over the first \(1\,\mu\mathrm{s}\). The red
\(\lambda_{\mathrm{G,max}}=3.14\times10^7\,\mathrm{s}^{-1}\) reference
resolves the fast edge. The black \(L_{\mathrm G}\) reference remains near
its initial level.}
\label{fig:heat-closeup}
\end{figure}

\begin{figure*}[!t]
\centering
\subfloat[CMOS steady profiles]{%
\includegraphics[width=.42\textwidth]{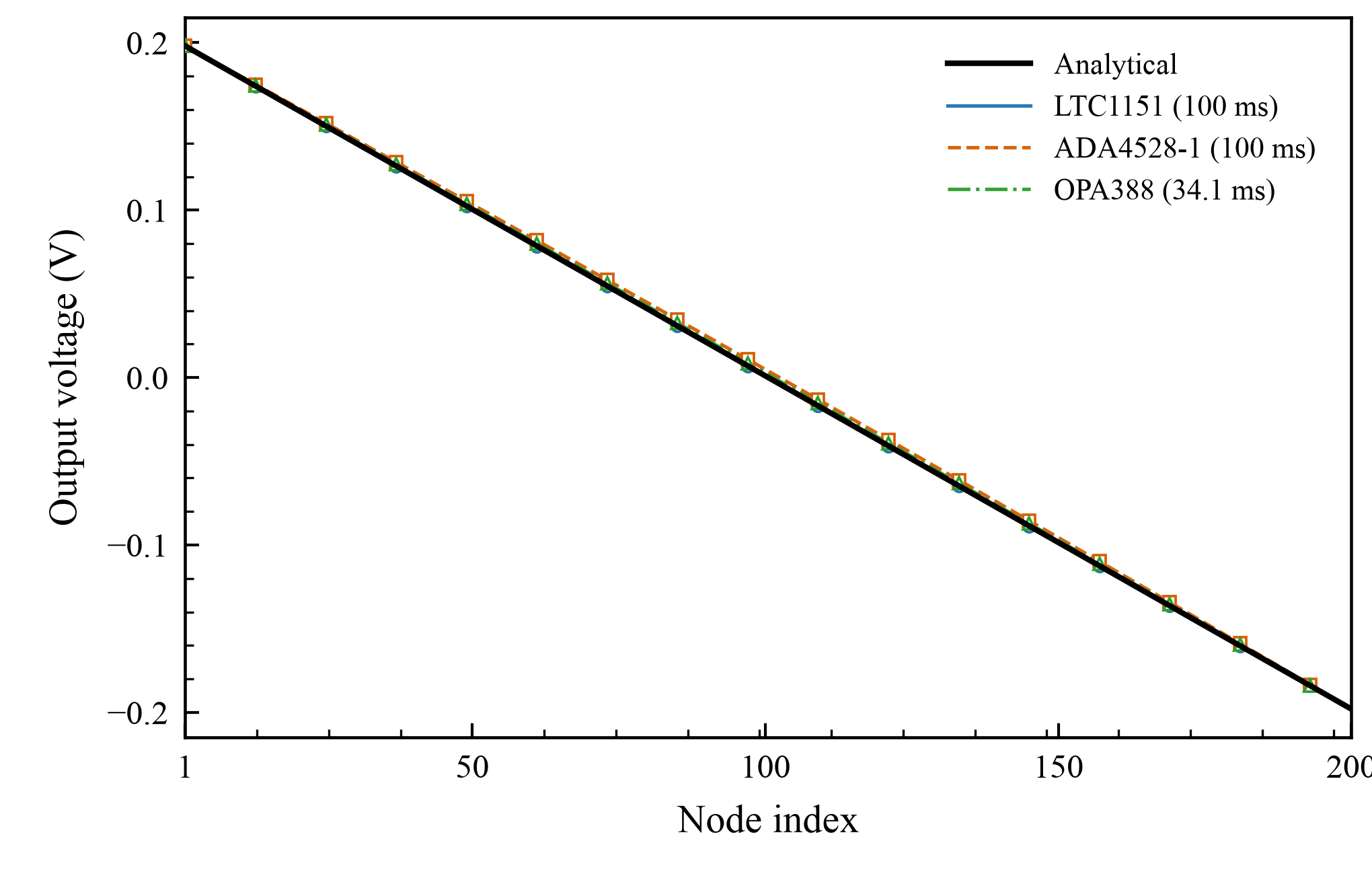}}
\hfill
\subfloat[Hybrid K2-W, \(0\)--\(300\,\mathrm{ms}\)]{%
\includegraphics[width=.42\textwidth]{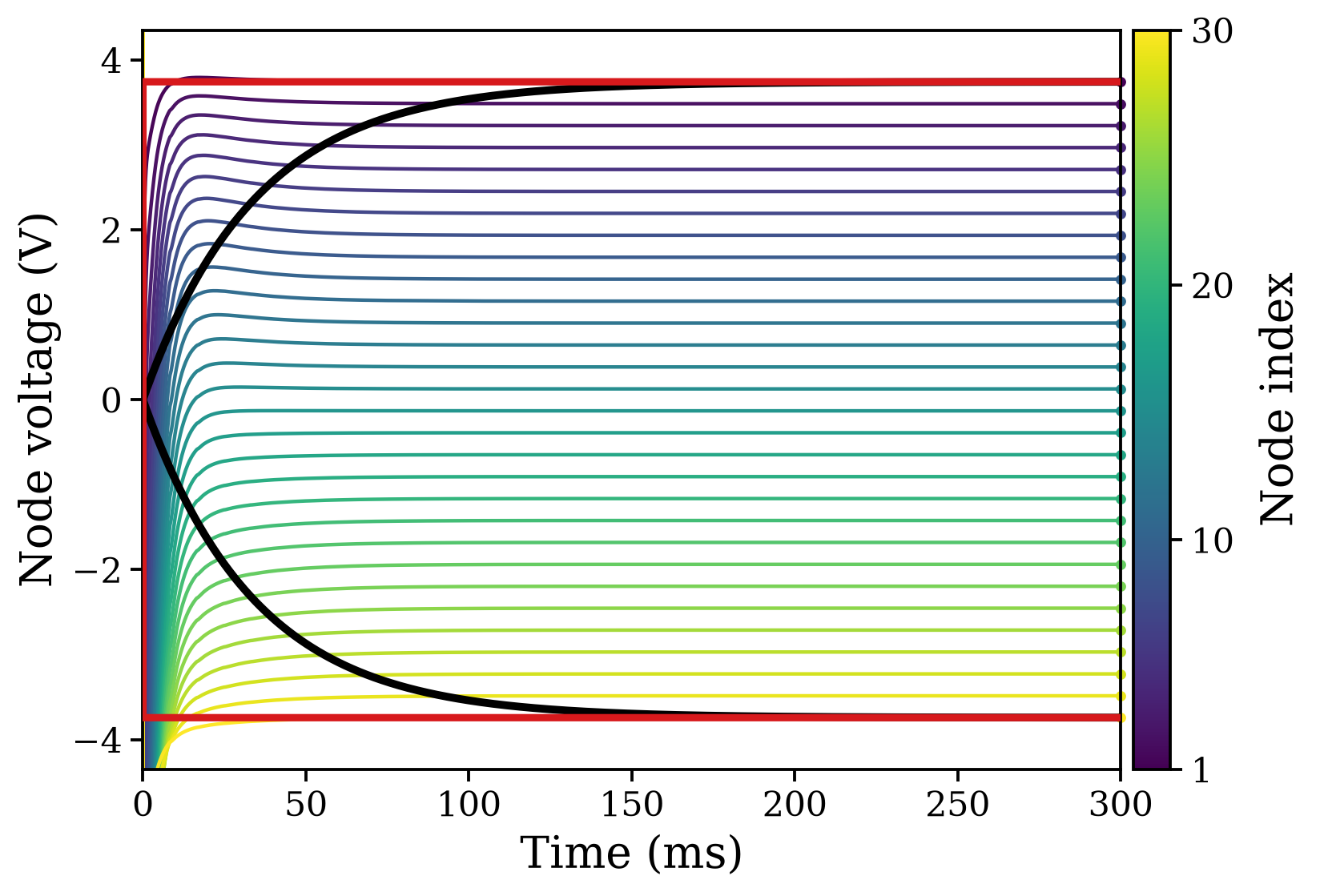}}
\caption{Additional heat-network macromodel evidence. (a) Settled CMOS
profiles compared with the ideal line. (b) Thirty-node hybrid
thermionic/behavioral transient; its thick reference curves indicate
reduced-model rate scales rather than pointwise bounds.}
\label{fig:heat-validation}
\end{figure*}

\begin{table*}[!t]
\centering
\caption{Generalized spectral endpoints and row-local upper certificates.
Heat-CMOS values use the closed form \eqref{eq:heat-spectrum}; graph values
use direct eigensolutions of the assembled matrices. The K2-W row uses the
fitted one-pole surrogate described in the text.}
\label{tab:spectrum}
\setlength{\tabcolsep}{5.0pt}
\begin{tabular}{lrrrrrrr}
\toprule
Case & \(n\) & \(\lambda_{\min}\) (s\(^{-1}\)) &
\(\lambda_{\max}\) (s\(^{-1}\)) &
\(\lambda_{\mathrm{G,max}}\) (s\(^{-1}\)) &
\(\lambda_{\mathrm{G,max}}/\lambda_{\max}\) &
\(\kappa_{\mathrm{dyn}}\) &
\(\lambda_{\mathrm{G,max}}/2\) \\
\midrule
Heat, LTC1151 & 200 & \(1.920{\times}10^3\) & \(3.1414{\times}10^7\) &
\(3.1416{\times}10^7\) & 1.00006 & \(1.636{\times}10^4\) & 15.708 MHz\\
Heat, ADA4528-1 & 200 & \(3.078{\times}10^3\) & \(5.0262{\times}10^7\) &
\(5.0265{\times}10^7\) & 1.00006 & \(1.633{\times}10^4\) & 25.133 MHz\\
Heat, OPA388 & 200 & \(7.677{\times}10^3\) & \(1.2566{\times}10^8\) &
\(1.2566{\times}10^8\) & 1.00006 & \(1.637{\times}10^4\) & 62.832 MHz\\
Heat, K2-W surrogate & 30 & \(3.38{\times}10^4\) & \(1.314{\times}10^7\) &
\(1.32{\times}10^7\) & 1.0046 & 388.8 & 6.60 MHz\\
Graph labeling & 200 & \(1.67{\times}10^5\) & \(1.98{\times}10^7\) &
\(3.31{\times}10^7\) & 1.672 & 118.6 & 16.55 MHz\\
13-band regression & 200 & \(2.37{\times}10^6\) & \(1.89{\times}10^7\) &
\(2.95{\times}10^7\) & 1.561 & 7.975 & 14.75 MHz\\
\bottomrule
\end{tabular}
\end{table*}

\begin{figure*}[!t]
\centering
\subfloat[Graph classification]{%
\includegraphics[width=.44\textwidth]{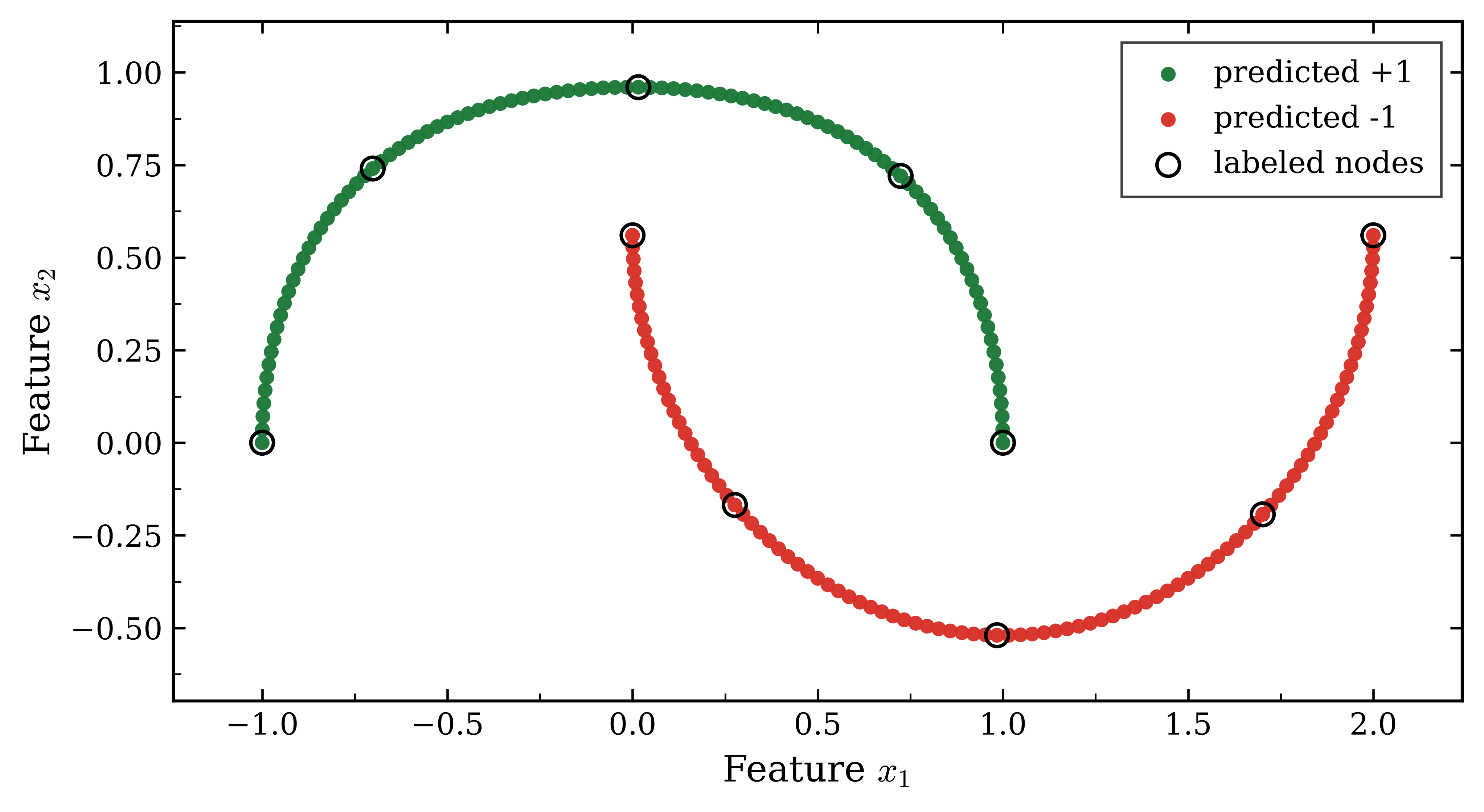}}
\hfill
\subfloat[Graph settling and bias]{%
\includegraphics[width=.44\textwidth]{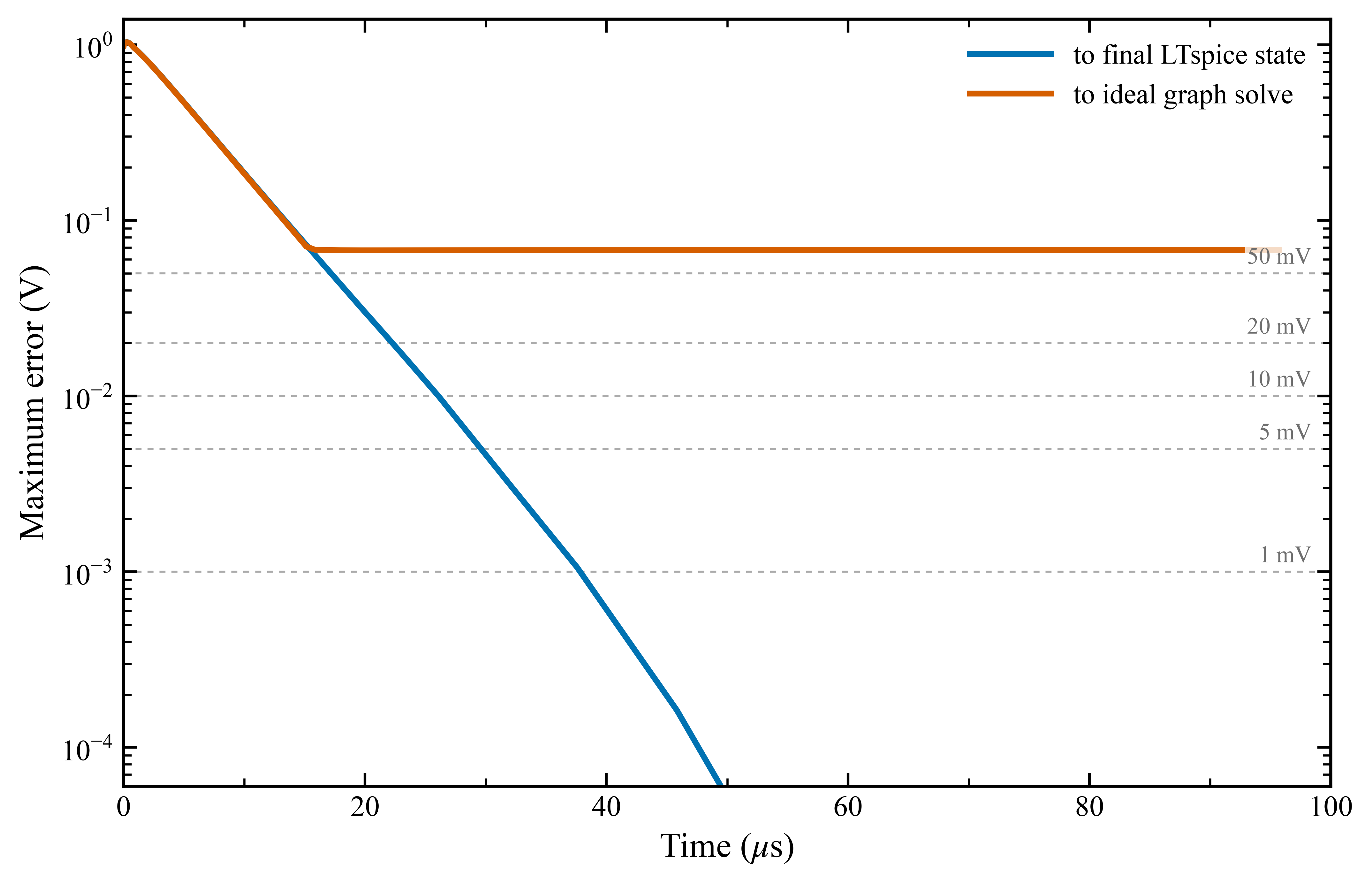}}\\[-.2em]
\subfloat[Regression, \(0\)--\(5\,\mu\mathrm{s}\)]{%
\includegraphics[width=.44\textwidth]{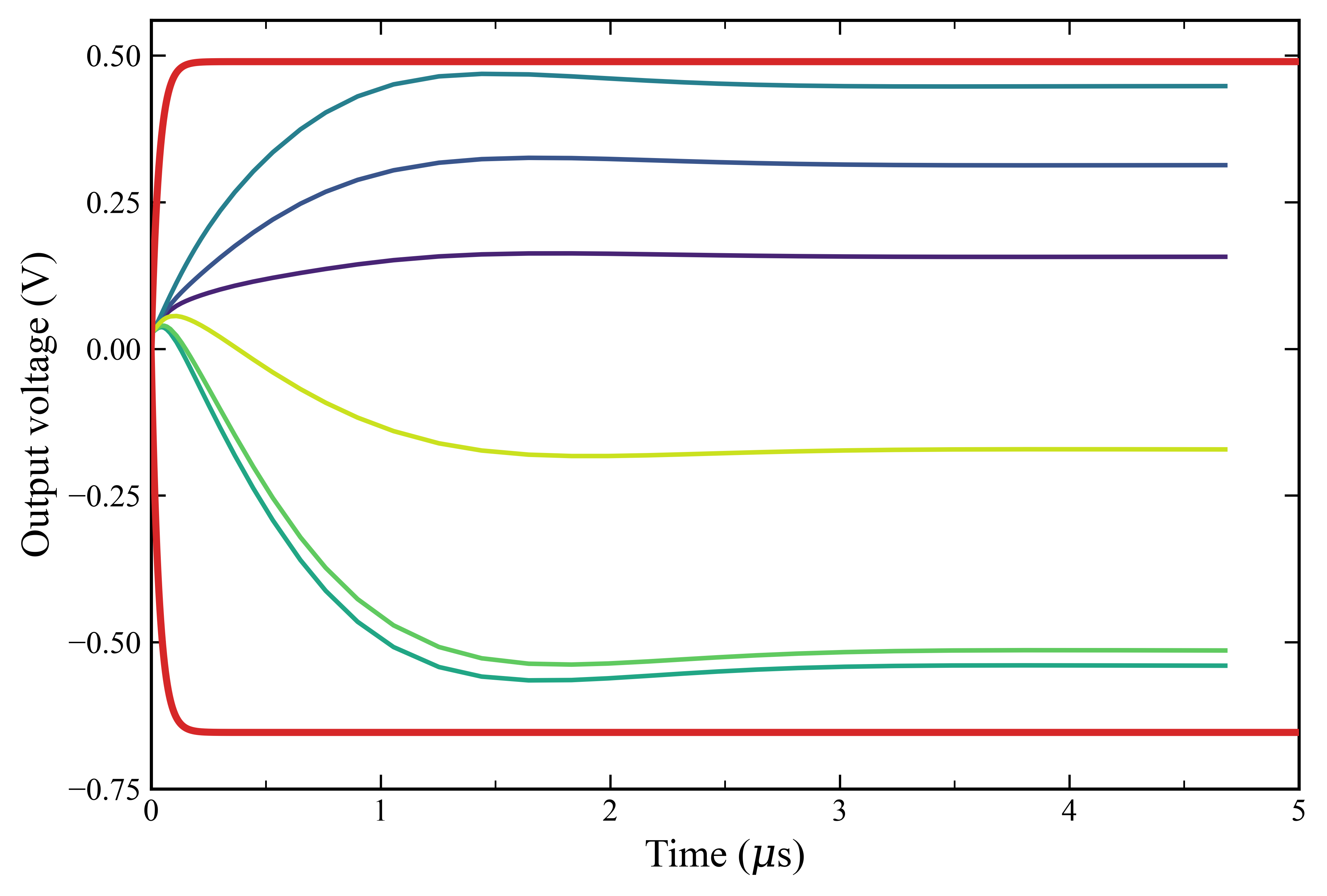}}
\hfill
\subfloat[Regression fixed point]{%
\includegraphics[width=.44\textwidth]{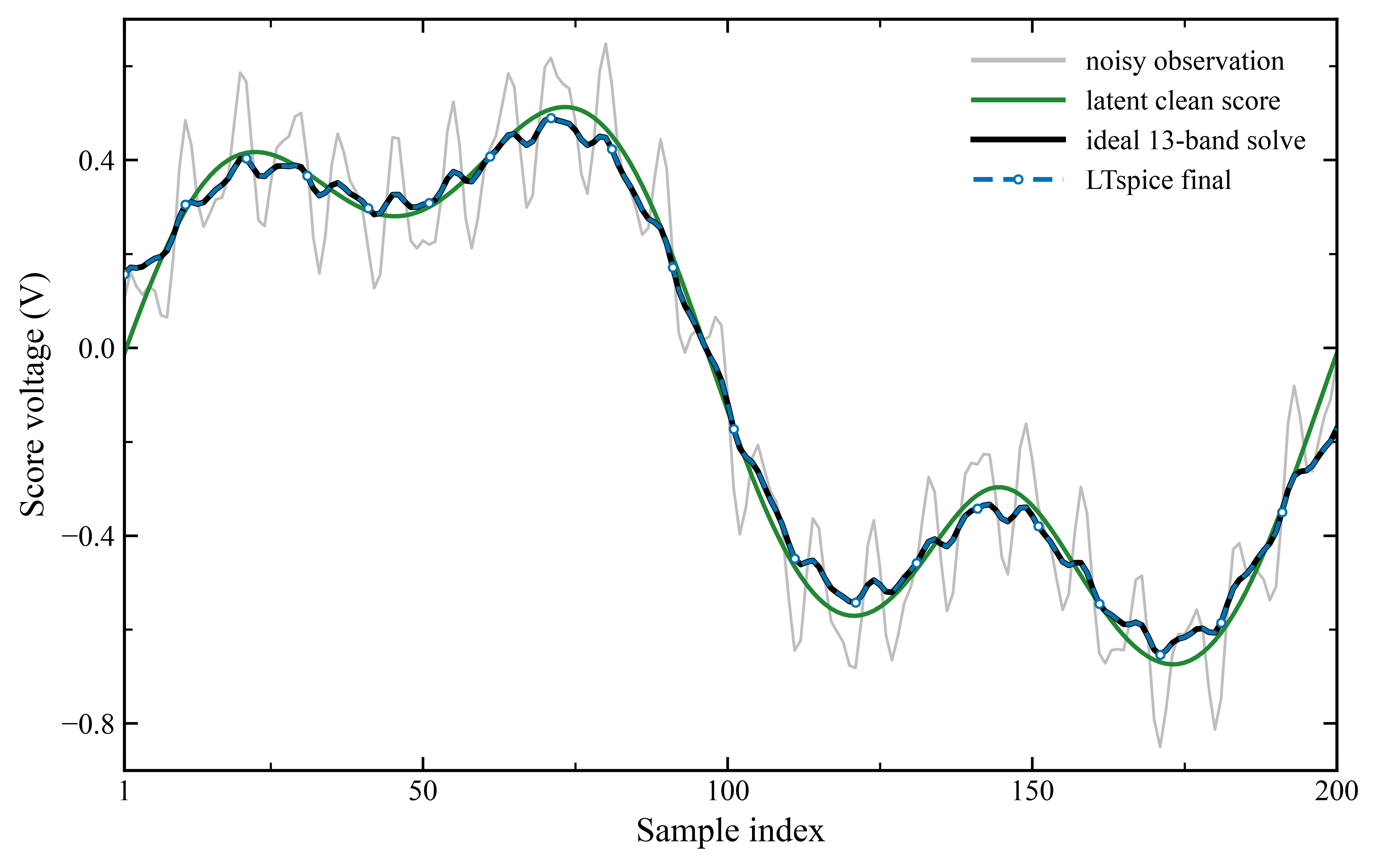}}
\caption{Irregular-topology macromodel studies. (a) Correct two-moons
classification from ten labeled nodes. (b) Error to the final circuit state
and to the ideal graph solve. (c) Representative 13-band regression
transients; the red scale marker is not a voltage bound. (d) Final regression
voltages overlap the ideal solution.}
\label{fig:graph-validation}
\end{figure*}

\begin{table}[!t]
\centering
\caption{Application-Level LTspice Outcomes}
\label{tab:application}
\setlength{\tabcolsep}{3.2pt}
\begin{tabular}{lrr}
\toprule
Case & Fixed-point result & Tolerance/task time\\
\midrule
Heat LTC1151 & \(0.040\) mV max. & \(1\) mV: \(2.62\) ms\\
Heat ADA4528-1 & \(3.84\) mV max. & \(4\) mV: \(20.53\) ms\\
Heat OPA388 & \(1.72\) mV max. & \(2\) mV: \(1.35\) ms\\
Heat K2-W hybrid & \(2.08\) mV max. & \(4\) mV: \(107.4\) ms\\
Graph labels & 200/200 correct & all signs: \(1.29\,\mu\)s\\
13-band regression & \(0.394\,\mu\)V max. & \(1\) mV: \(3.49\,\mu\)s\\
\bottomrule
\end{tabular}
\end{table}

Fig.~\ref{fig:certificate-euler} summarizes the upper-edge validation. Panel
(a) compares the row-local certificates with the exact upper rates across the
six cases, while panel (b) illustrates the three fastest-mode update regimes.

\begin{figure*}[!t]
\centering
\resizebox{\textwidth}{!}{
\begin{tikzpicture}[x=1cm,y=1cm,font=\scriptsize]
\begin{scope}[shift={(0.55,0.55)}]
\draw[->] (0,0) -- (7.10,0);
\draw[->] (0,0) -- (0,4.45);
\node[rotate=90] at (-0.62,2.2) {$\lambda_{\mathrm{G,max}}/\lambda_{\max}$};
\foreach \v/\yy in {1.0/0,1.2/1.0,1.4/2.0,1.6/3.0,1.8/4.0}{
  \draw (-0.07,\yy) -- (0.07,\yy); \node[left] at (-0.09,\yy) {\v};
  \draw[gray!25] (0,\yy) -- (6.75,\yy);
}
\fill[blue!65] (0.470,0) rectangle (1.030,0.0003);
\node[above] at (0.750,0.0003) {1.000};
\node[align=center,anchor=north] at (0.750,-0.13) {LTC};
\fill[blue!65] (1.470,0) rectangle (2.030,0.0003);
\node[above] at (1.750,0.0003) {1.000};
\node[align=center,anchor=north] at (1.750,-0.13) {ADA};
\fill[blue!65] (2.470,0) rectangle (3.030,0.0003);
\node[above] at (2.750,0.0003) {1.000};
\node[align=center,anchor=north] at (2.750,-0.13) {OPA};
\fill[blue!65] (3.470,0) rectangle (4.030,0.0228);
\node[above] at (3.750,0.0228) {1.005};
\node[align=center,anchor=north] at (3.750,-0.13) {K2-W};
\fill[blue!65] (4.470,0) rectangle (5.030,3.3586);
\node[above] at (4.750,3.3586) {1.67};
\node[align=center,anchor=north] at (4.750,-0.13) {Graph};
\fill[blue!65] (5.470,0) rectangle (6.030,2.8042);
\node[above] at (5.750,2.8042) {1.56};
\node[align=center,anchor=north] at (5.750,-0.13) {13-band};
\node[anchor=north west,font=\bfseries] at (0.02,4.38) {(a)};
\end{scope}
\begin{scope}[shift={(8.35,2.72)}]
\clip (-0.15,-2.25) rectangle (7.0,2.25);
\draw[->] (0,-2.12) -- (0,2.18);
\draw[->] (0,0) -- (6.75,0);
\foreach \n/\xx in {0/0,4/1.6,8/3.2,12/4.8,16/6.4}{\draw (\xx,-0.06)--(\xx,0.06); \node[below] at (\xx,-0.08) {\n};}
\foreach \v/\yy in {-2/-1.8,-1/-0.9,0/0,1/0.9,2/1.8}{\draw (-0.06,\yy)--(0.06,\yy); \node[left] at (-0.08,\yy) {\v};}
\foreach \yy in {-1.8,-0.9,0.9,1.8}{\draw[gray!25] (0,\yy)--(6.55,\yy);}
\node[rotate=90] at (-0.70,0) {Fastest-mode amplitude};
\node at (3.25,-2.05) {Explicit-update index, $n$};
\node[anchor=north west,font=\bfseries] at (0.08,2.10) {(b)};
\draw[green!55!black,thick,mark=*,mark size=0.65pt] plot coordinates {(0.000,0.90000) (0.400,0.18000) (0.800,0.03600) (1.200,0.00720) (1.600,0.00144) (2.000,0.00029) (2.400,0.00006) (2.800,0.00001) (3.200,0.00000) (3.600,0.00000) (4.000,0.00000) (4.400,0.00000) (4.800,0.00000) (5.200,0.00000) (5.600,0.00000) (6.000,0.00000) (6.400,0.00000)};
\draw[violet!85!black,thick,mark=*,mark size=0.65pt] plot coordinates {(0.000,0.90000) (0.400,-0.18000) (0.800,0.03600) (1.200,-0.00720) (1.600,0.00144) (2.000,-0.00029) (2.400,0.00006) (2.800,-0.00001) (3.200,0.00000) (3.600,-0.00000) (4.000,0.00000) (4.400,-0.00000) (4.800,0.00000) (5.200,-0.00000) (5.600,0.00000) (6.000,-0.00000) (6.400,0.00000)};
\draw[orange!90!black,thick,mark=*,mark size=0.65pt] plot coordinates {(0.000,0.90000) (0.400,-0.81000) (0.800,0.72900) (1.200,-0.65610) (1.600,0.59049) (2.000,-0.53144) (2.400,0.47830) (2.800,-0.43047) (3.200,0.38742) (3.600,-0.34868) (4.000,0.31381) (4.400,-0.28243) (4.800,0.25419) (5.200,-0.22877) (5.600,0.20589) (6.000,-0.18530) (6.400,0.16677)};
\draw[red!75!black,thick,mark=*,mark size=0.65pt] plot coordinates {(0.000,0.90000) (0.400,-0.94500) (0.800,0.99225) (1.200,-1.04186) (1.600,1.09396) (2.000,-1.14865) (2.400,1.20609) (2.800,-1.26639) (3.200,1.32971) (3.600,-1.39620) (4.000,1.46601) (4.400,-1.53931) (4.800,1.61627) (5.200,-1.69708) (5.600,1.78194) (6.000,-1.87104) (6.400,1.96459)};
\end{scope}
\begin{scope}[shift={(8.72,4.48)}]
\draw[green!55!black,thick] (0,0)--(0.38,0); \node[right] at (0.42,0) {$T_s\lambda_{\max}=0.8$};
\draw[violet!85!black,thick] (2.15,0)--(2.53,0); \node[right] at (2.57,0) {$1.2$};
\draw[orange!90!black,thick] (3.55,0)--(3.93,0); \node[right] at (3.97,0) {$1.9$};
\draw[red!75!black,thick] (4.95,0)--(5.33,0); \node[right] at (5.37,0) {$2.05$};
\end{scope}
\end{tikzpicture}}
\caption{Certificate and sampled-update behavior. (a) Ratio of the row-local
upper certificate to the exact generalized upper rate. (b) Exact fastest-mode
recurrence \((1-\CoE)^n\), where \(0<\CoE\le1\) is nonalternating,
\(1<\CoE<2\) is alternating but stable, and \(\CoE>2\) is unstable.}
\label{fig:certificate-euler}
\end{figure*}
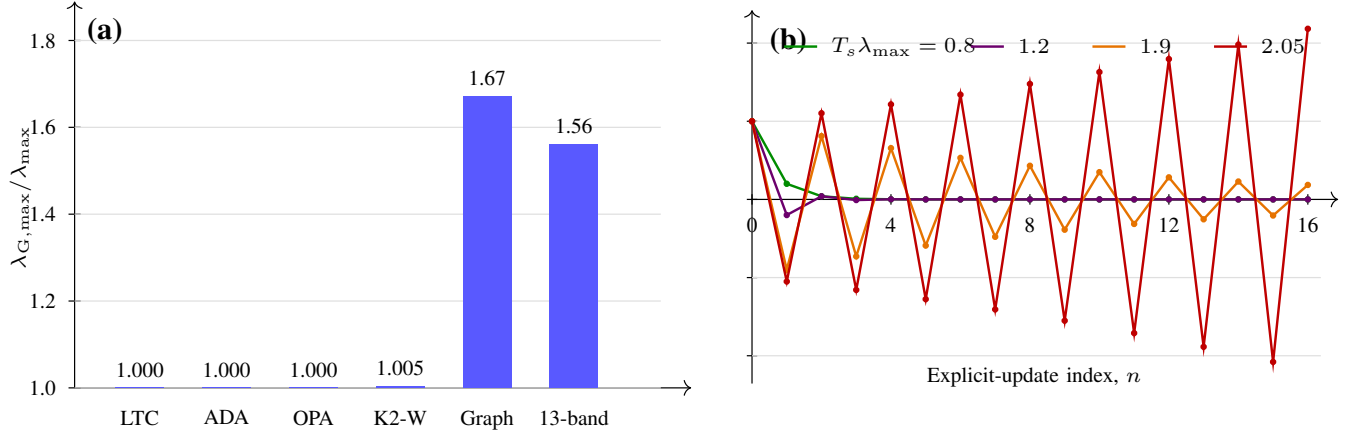

\newpage
\section{Design Implications and Limitations}
The certificate is most useful before layout or full-array pole extraction.
It exposes which rows combine large \(\nu_t\), weak loading \(S_j\), and
strong normalized coupling. Those rows set the acquisition and explicit
update requirements. Equivalent variable scaling or preconditioning can
reduce \(\kappa_{\mathrm{dyn}}\) without changing the recovered solution
\cite{feinberg2021preconditioner}, and the bandwidth can then be allocated
heterogeneously using \eqref{eq:row-design}. Stronger anchors improve the
spectral gap only if they belong to the intended objective, because otherwise
they change the solved matrix. Uniform resistance scaling alone does not
change the reduced rate.

The analysis assumes reciprocal nonnegative couplings, positive diagonal
\(\mat C\), and small-signal one-pole follower dynamics
\cite{lin1986}. Active, signed, or nonnormal networks require a
realization-specific Lyapunov or pseudospectral analysis. The certificate
bounds the generalized eigenvalues of the reduced model, not every internal
pole of a vendor or chopper macromodel. Slew, saturation, additional poles,
mismatch, PVT variation, noise, converters, and interconnect can introduce
slower or nonmonotone behavior. The rail test is sufficient and conservative,
and its failure does not imply saturation. The evidence here is
macromodel-based rather than silicon measurement. The thermionic example utilizes a
hybrid tube-level/behavioral. Finally, the sampled limits apply only when a
controller implements \eqref{eq:euler-physical}. They are not stability
requirements for passive observation.

\section{Conclusion}
Finite amplifier bandwidth maps programmed row loading into a diagonal
dynamic metric. The resulting generalized spectrum has two distinct design
scales. The rate \(\lambda_{\min}\) governs worst-case settling, whereas
\(\lambda_{\max}\) governs the fastest reduced-model mode and the explicit
sampled-update limit. The proposed row-local certificate bounds the latter
without a global eigensolution. Its continuous-versus-sampled interpretation, conservative rail
condition, direct modal verification, and cross-topology results provide a
practical basis for choosing amplifier bandwidth, sampling rate, and network
conditioning while keeping the reduced model's assumptions explicit.

\appendices
\section{Hybrid Stabilized K2-W/K2-P Macromodel}
\label{app:k2w-macro}
\begin{figure*}[!t]
\centering
\subfloat[Assembled wrapper and expanded behavioral K2-P stabilization path.]{%
\resizebox{.98\textwidth}{!}{%
\begin{circuitikz}[american voltages, every node/.style={font=\scriptsize}]
\ctikzset{bipoles/length=0.82cm}
\tikzset{
  signal/.style={line width=0.55pt, line cap=round, line join=round},
  heater/.style={line width=0.55pt, green!45!black},
  rail/.style={line width=0.8pt},
  block/.style={draw, rounded corners=2pt, fill=blue!4, align=center,
    minimum width=1.65cm, minimum height=0.68cm, font=\scriptsize},
  sourceblock/.style={draw, rounded corners=2pt, fill=orange!9, align=center,
    minimum width=1.82cm, minimum height=0.72cm, font=\scriptsize},
  port/.style={font=\scriptsize\ttfamily, fill=white, inner sep=1pt},
  note/.style={font=\scriptsize, fill=white, inner sep=1pt}
}

\draw[rail] (-0.25,9.05) -- (21.25,9.05)
  node[right,note]{$V_+=+300\,\mathrm V$};
\draw[rail] (-0.25,0) -- (21.25,0) node[right,note]{COM};

\draw[signal] (0,7.45) node[ocirc]{} node[left,port]{INP}
  -- (2.10,7.45);
\draw[signal] (0,6.25) node[ocirc]{} node[left,port]{INN}
  -- (2.10,6.25);
\draw[signal] (0.50,7.45) to[R] (0.50,8.40)
  node[ocirc]{} node[right,port]{COM};
\draw[signal] (1.65,6.25) to[R] (1.65,5.15)
  node[ocirc]{} node[below,port]{COM};
\node[note, anchor=east, align=right] at (0.42,7.92)
  {$R_{\mathrm{INP,REF}}$\\$1\,\mathrm{T}\Omega$};
\node[note, anchor=west, align=left] at (1.73,5.70)
  {$R_{\mathrm{INN,REF}}$\\$1\,\mathrm{T}\Omega$};

\node[sourceblock] (BERR) at (3.25,6.85)
  {BERR\\$V(\mathrm{INP},\mathrm{INN})$};
\draw[signal] (2.10,7.45) |- ([yshift=0.16cm]BERR.west);
\draw[signal] (2.10,6.25) |- ([yshift=-0.16cm]BERR.west);
\draw[signal] (BERR.south) -- (3.25,5.75)
  node[ocirc]{} node[below,port]{COM};

\node[block] (K2P) at (6.20,6.85)
  {STAB\_K2P\_CORE\\expanded below};
\draw[signal] (BERR.east) -- (K2P.west)
  node[midway,above,port]{PERR};
\node[note, below=0.05cm of K2P] {$V_+,\mathrm{COM},\mathrm{H1},\mathrm{H2}$};

\node[sourceblock, minimum width=2.45cm] (BREF) at (10.00,6.85)
  {BREF\\$\mathrm{INP}+V_{\mathrm{OS}}$\\$+\mathrm{PCORR}-\mathrm{BIAS}$};
\draw[signal] (K2P.east) -- (BREF.west)
  node[midway,above,port]{PCORR};
\draw[signal] (2.10,7.45) -- (2.10,8.35) -- (10.00,8.35)
  -- (BREF.north);
\draw[signal] (BREF.south) -- (10.00,5.72)
  node[ocirc]{} node[below,port]{COM};

\node[block, minimum width=2.40cm, minimum height=1.08cm] (K2W)
  at (14.05,6.85) {STAB\_K2W\_CORE\\tube-level core\\shown in (b)};
\coordinate (K2W_INP_PORT) at ([yshift=0.20cm]K2W.west);
\coordinate (K2W_INN_PORT) at ([yshift=-0.20cm]K2W.west);
\draw[signal] (BREF.east)
  -- node[midway,above,port]{K2W\_INP} (12.45,6.85)
  |- (K2W_INP_PORT);
\draw[signal] (2.10,6.25) -- (2.10,5.45)
  -- ([yshift=-1.20cm]K2W_INN_PORT)
  -- node[midway,left,port]{INN} (K2W_INN_PORT);
\draw[signal] (K2W.east) -- (16.55,6.85) node[ocirc]{}
  node[right,port]{OUT};
\node[note, below=0.05cm of K2W] {$V_+,V_-,\mathrm{COM}$};
\node[note, align=left, anchor=west] at (17.05,7.15)
  {$V_{\mathrm{OS}}=0$\\$\mathrm{BIAS}=-1.705\,\mathrm V$\\$G_{\mathrm{K2P}}=1000$\\$V_{\mathrm{OUTLIM}}=30\,\mathrm V$\\\texttt{DRIVE\_SHARPNESS}=2};

\node[draw, dashed, rounded corners=2pt, fit={(0.05,0.30) (20.85,4.75)},
  inner sep=4pt, label={[note]above:expanded \texttt{STAB\_K2P\_CORE}}] {};
\draw[signal] (0.35,2.80) node[ocirc]{} node[left,port]{PERR}
  -- (2.10,2.80);
\draw[signal] (1.10,2.80) to[R,l_=$R_{\mathrm{IN}}\ 2\,\mathrm{M}\Omega$]
  (1.10,0);

\node[block] (BCHOP) at (3.10,2.80) {BCHOP\\chop};
\node[block] (BAMP) at (7.00,2.80) {BAMP\\gain/limit};
\node[block] (BDEMOD) at (10.90,2.80) {BDEMOD\\demod};
\draw[signal] (2.10,2.80) -- (BCHOP.west);
\draw[signal] (BCHOP.east) -- (BAMP.west)
  node[midway,above,port]{CHOP};
\draw[signal] (BAMP.east) -- (BDEMOD.west)
  node[midway,above,port]{AMP};
\draw[signal] (BDEMOD.east) -- (12.35,2.80)
  node[midway,above,port]{DET}
  to[R,l=$R_{\mathrm{DET}}\ 22\,\mathrm{M}\Omega$] (14.65,2.80)
  -- (15.50,2.80) node[ocirc]{} node[right,port]{PCORR};
\draw[signal] (14.65,2.80) to[C,l_=$C_{\mathrm{OUT}}\ 1\,\mu\mathrm F$,a^={$R_s=1\,\Omega$}]
  (14.65,0);

\node[block, minimum width=1.55cm] (BPWR) at (7.60,4.12)
  {BPWR\\supply gate\\$V(V_+,\mathrm{COM})$};
\draw[signal] (7.60,9.05) -- (7.60,8.55)
  arc[start angle=90,end angle=-90,radius=0.20]
  -- (7.60,7.05)
  arc[start angle=90,end angle=-90,radius=0.20]
  -- (7.60,5.65)
  arc[start angle=90,end angle=-90,radius=0.20]
  -- (BPWR.north);
\draw[signal] (BPWR.south) -- (7.60,3.38)
  -- node[midway,above,port]{PWR\_OK} (7.00,3.38)
  -- (BAMP.north);
\draw[signal] (17.05,9.05) to[R,l=$R_{\mathrm{HV}}\ 125\,\mathrm{k}\Omega$]
  (17.05,0);

\draw[heater] (2.10,1.20) node[ocirc]{} node[left,port]{H1}
  to[R,l=$R_{\mathrm{HEATER}}\ 14\,\Omega$] (4.65,1.20)
  node[ocirc]{} node[right,port]{H2};
\draw[heater] (4.65,1.20) to[R,l_=$R_{\mathrm{HREF}}\ 1\,\mathrm{T}\Omega$]
  (4.65,0);
\draw[heater,->] (2.55,1.20) |- (BCHOP.south);
\draw[heater,->] (4.20,1.20) -| (BDEMOD.south);
\end{circuitikz}}}
\par\vspace{.35em}
\subfloat[Tube-level K2-W core and passive compensation network.]{%
\resizebox{.98\textwidth}{!}{%
\begin{circuitikz}[american voltages, every node/.style={font=\scriptsize}]
\ctikzset{
  tubes/width=0.88,
  tubes/height=1.20,
  bipoles/length=0.86cm
}
\tikzset{
  signal/.style={line width=0.55pt, line cap=round, line join=round},
  railplus/.style={line width=0.85pt, red!65!black},
  railminus/.style={line width=0.85pt, blue!65!black},
  railcom/.style={line width=0.75pt, gray!70!black},
  port/.style={font=\scriptsize\ttfamily, fill=white, inner sep=1pt},
  note/.style={font=\scriptsize, fill=white, inner sep=1pt},
  tubelabel/.style={font=\scriptsize, align=center, fill=white, inner sep=1pt},
  neon/.style={draw, circle, fill=yellow!15, minimum size=0.48cm,
    inner sep=0pt, font=\tiny}
}

\draw[railplus] (-0.15,7.10) -- (22.75,7.10)
  node[right,note]{$V_+=+300\,\mathrm V$};
\draw[railcom] (-0.15,1.15) -- (22.75,1.15) node[right,note]{COM};
\draw[railminus] (-0.15,-0.15) -- (22.75,-0.15)
  node[right,note]{$V_-=-300\,\mathrm V$};

\node[triode] (U1) at (3.00,5.00) {};
\node[tubelabel, above] at (U1.north) {XU1\\12AX7};
\node[triode] (U2) at (6.10,4.00) {};
\node[tubelabel, above] at (U2.north) {XU2\\12AX7};
\node[triode] (U3) at (12.15,4.00) {};
\node[tubelabel, above] at (U3.north) {XU3\\12AX7};
\node[triode] (U4) at (19.15,4.00) {};
\node[tubelabel, above] at (U4.north) {XU4\\12AX7};

\draw[signal] (0.20,5.00) node[ocirc]{} node[left,port]{INN}
  -- (U1.grid);
\draw[signal] (0.20,4.00) node[ocirc]{} node[left,port]{K2W\_INP}
  -- (U2.grid);
\draw[signal] (U1.anode) -- (3.00,7.10);
\coordinate (K12) at (4.55,2.25);
\draw[signal] (U1.cathode) |- (K12);
\draw[signal] (U2.cathode) |- (K12);
\node[circ] at (K12) {};
\node[above,port] at (K12) {K12};
\draw[signal] (K12) -- (4.55,1.25)
  arc[start angle=90,end angle=-90,radius=0.10]
  -- (4.55,0.95)
  to[R,l=$R_5\ 220\,\mathrm{k}\Omega$] (4.55,0.05)
  -- (4.55,-0.15);

\coordinate (P2) at (6.10,5.50);
\draw[signal] (U2.anode) -- (P2) node[above,port]{P2};
\node[circ] at (P2) {};
\draw[signal] (P2) to[R,l=$R_1\ 220\,\mathrm{k}\Omega$]
  (6.10,7.10);
\coordinate (G3) at (10.35,4.00);
\draw[signal] (P2) -- (8.00,5.50)
  to[R,l=$R_4\ 1\,\mathrm{M}\Omega$] (10.35,5.50)
  -- (G3) -- (U3.grid);
\node[circ] at (G3) {};
\node[above,port] at (G3) {G3};
\draw[signal] (G3) -- (10.35,1.25)
  arc[start angle=90,end angle=-90,radius=0.10]
  -- (10.35,0.95)
  to[R,l_=$R_6\ 2.2\,\mathrm{M}\Omega$] (10.35,0.05)
  -- (10.35,-0.15);

\coordinate (P3) at (12.15,5.50);
\draw[signal] (U3.anode) -- (P3) node[above left,port]{P3};
\node[circ] at (P3) {};
\draw[signal] (P3) to[R,l=$R_2\ 470\,\mathrm{k}\Omega$]
  (12.15,7.10);
\draw[signal] (P3) -- (13.05,5.50)
  to[R,l=$R_3\ 680\,\mathrm{k}\Omega$] (15.40,5.50)
  coordinate (NE2A) node[above,port]{NE2A};
\node[neon] (NE1) at (16.10,5.50) {XNE1};
\node[neon] (NE2) at (17.10,5.50) {XNE2};
\node[port] at (16.60,5.18) {NE2B};
\coordinate (G4) at (18.35,5.50);
\draw[signal] (NE2A) -- (NE1.west);
\draw[signal] (NE1.east) -- (NE2.west);
\draw[signal] (NE2.east) -- (G4) -- (18.35,4.00) -- (U4.grid);
\node[above,port] at (G4) {G4};
\draw[signal] (P3) -- (13.05,5.50) -- (13.05,4.72)
  -- (13.75,4.72)
  to[C,l_=$C_3\ 7.5\,\mathrm{pF}$] (15.35,4.72)
  -- (18.35,4.72) node[circ]{};
\draw[signal] (G4) -- (18.35,2.35)
  arc[start angle=90,end angle=-90,radius=0.10]
  -- (18.35,1.25)
  arc[start angle=90,end angle=-90,radius=0.10]
  -- (18.35,0.95)
  to[R,l_=$R_8\ 4.7\,\mathrm{M}\Omega$] (18.35,0.05)
  -- (18.35,-0.15);

\coordinate (K3) at (12.15,2.25);
\draw[signal] (U3.cathode) |- (K3) node[left,port]{K3};
\node[circ] at (K3) {};
\draw[signal] (K3) to[R,l_=$R_7\ 10\,\mathrm{k}\Omega$]
  (12.15,1.15) node[circ]{};
\draw[signal] (K3) -- (13.65,2.25)
  to[C] (13.65,1.15) node[circ]{};
\node[note, anchor=west] at (14.05,1.84)
  {$C_1\ 500\,\mathrm{pF}$};

\coordinate (OUT) at (19.15,2.25);
\draw[signal] (U4.anode) -- (19.15,7.10);
\draw[signal] (U4.cathode) |- (OUT) -- (22.20,2.25)
  node[ocirc]{} node[right,port]{OUT};
\node[circ] at (OUT) {};
\draw[signal] (K3) to[R,l=$R_{11}\ 220\,\mathrm{k}\Omega$] (OUT);
\draw[signal] (10.35,3.05) node[circ]{}
  -- (11.97,3.05)
  arc[start angle=180,end angle=0,radius=0.18]
  -- (14.85,3.05)
  to[C,l=$C_2\ 7.5\,\mathrm{pF}$] (16.45,3.05)
  -- (17.10,3.05) -- (17.10,2.25) node[circ]{};
\draw[signal] (OUT) -- (20.10,2.25) node[circ]{} -- (20.10,1.25)
  arc[start angle=90,end angle=-90,radius=0.10]
  -- (20.10,0.95)
  to[R,l_=$R_9\ 220\,\mathrm{k}\Omega$] (20.10,0.05)
  -- (20.10,-0.15);
\draw[signal] (OUT) -- (21.10,2.25) node[circ]{} -- (21.10,1.25)
  arc[start angle=90,end angle=-90,radius=0.10]
  -- (21.10,0.95)
  to[R,l=$R_{10}\ 270\,\mathrm{k}\Omega$] (21.10,0.05)
  -- (21.10,-0.15);
\end{circuitikz}}}
\caption{Hybrid chopper-stabilized Philbrick K2-W macromodel used in the
30-node heat benchmark. The green connections in (a) identify the heater-driven chopping and
demodulation controls, and both represent dependence on the differential voltage
\(V(\mathrm{H1},\mathrm{H2})\), rather than additional electrical ports on the
signal chain. Every small bridge in (a) and (b) denotes a wire crossover, not
a junction. The \(C_1\)--\(C_3\) SPICE elements each use
\(R_{\mathrm{ser}}=1\,\Omega\).}
\label{fig:k2w-macro-schematic}
\end{figure*}
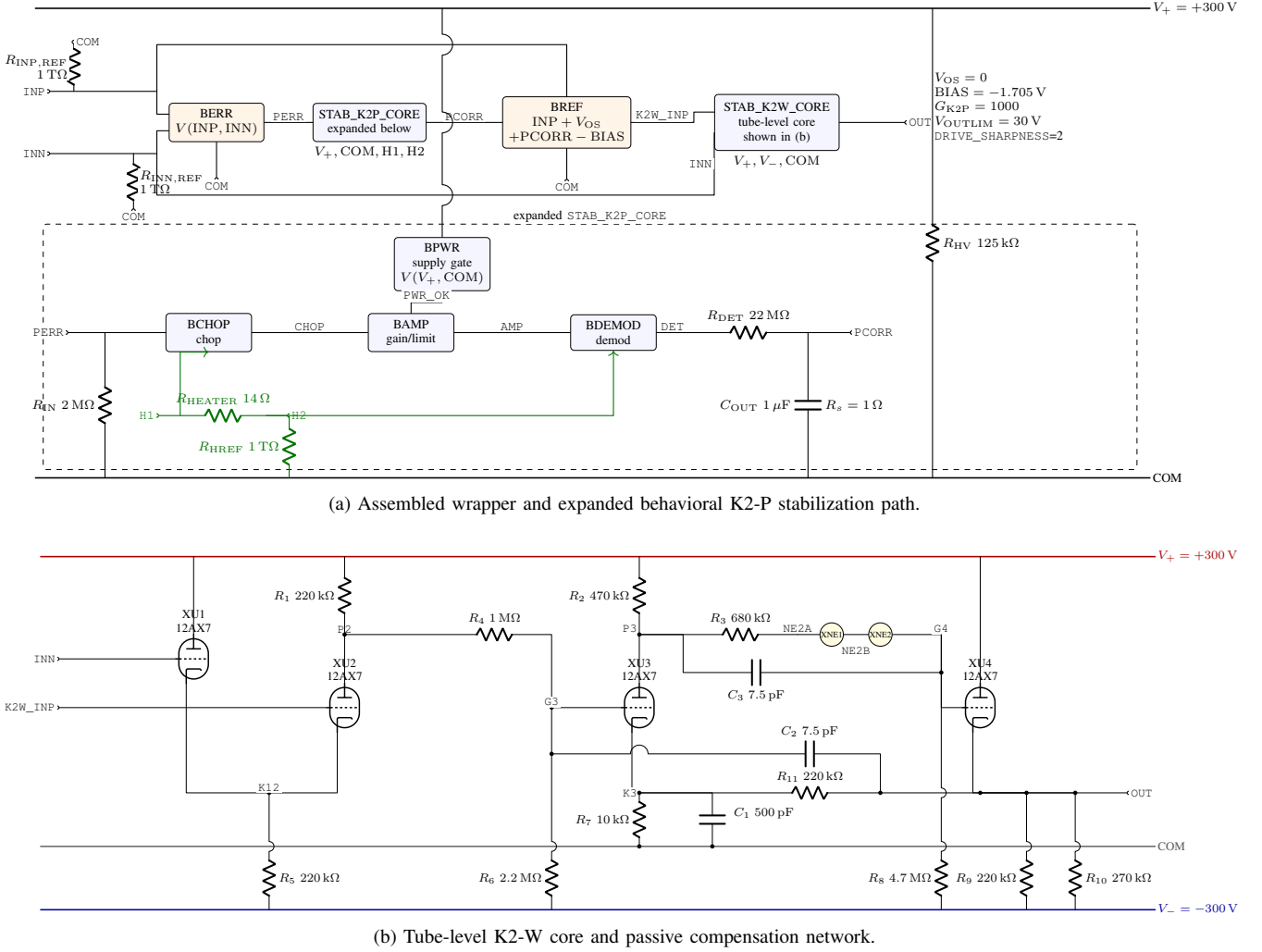

The thermionic vacuum-tube benchmark uses the assembled \texttt{STABILIZED\_K2W}
subcircuit summarized in Fig.~\ref{fig:k2w-macro-schematic}. The wrapper
senses the external differential error with \texttt{BERR}, processes its
low-frequency component through the K2-P chopper-stabilization core, and adds
the correction to the K2-W noninverting grid through \texttt{BREF}. The
benchmark closes unity feedback externally by connecting \texttt{OUT} to
\texttt{INN}, and this connection is not internal to the macro. The default model
parameters are \(V_{\mathrm{OS}}=0\),
\(\mathrm{BIAS}=-1.705\,\mathrm V\), K2-P gain \(1000\), and output limit
\(30\,\mathrm V\). Plate supplies are \(\pm300\,\mathrm V\), and the K2-P
heater/chopper terminals are driven by \(6.3\,\mathrm{V_{rms}}\) at
\(60\,\mathrm{Hz}\).

The K2-W signal path is represented at tube level by four 12AX7 triode
subcircuits, the passive bias and compensation network, and two neon-bulb
switch models. The K2-P path is intentionally behavioral, and voltage-defined
SPICE blocks implement its chopping, limited gain, synchronous demodulation,
and supply gating.


\IEEEtriggeratref{17}

\end{document}